\documentclass[a4paper,UKenglish,cleveref, autoref]{lipics-v2019}
\usepackage{makeidx}
\usepackage{amsmath,amssymb,amsfonts}
\usepackage{latexsym}
\usepackage{graphicx}
\usepackage{fancybox}
\usepackage{hyperref}
\usepackage{tikz}
\usetikzlibrary{decorations.pathreplacing}
\usetikzlibrary {positioning,chains,fit,shapes,calc,snakes}
\usepackage{standalone}

\tikzset{
    position/.style args={#1:#2 from #3}{
        at=(#3.#1), anchor=#1+180, shift=(#1:#2)
    }
}

\usepackage{setspace}
\usepackage{algorithm}
\usepackage[noend]{algpseudocode}
\usepackage[framemethod=tikz]{mdframed}
\usepackage{xspace}
\usepackage{pgfplots}
\usepackage{framed}
\pgfplotsset{compat=1.5}

\newtheorem{conjecture}[theorem]{Conjecture}

\newenvironment{proofof}[1]{\begin{trivlist} \item {\bf Proof
#1:~~}}
  {\qed\end{trivlist}}

\newcommand{\namedref}[2]{\hyperref[#2]{#1~\ref*{#2}}}
\newcommand{\thmlab}[1]{\label{thm:#1}}
\newcommand{\thmref}[1]{\namedref{Theorem}{thm:#1}}
\newcommand{\lemlab}[1]{\label{lem:#1}}
\newcommand{\lemref}[1]{\namedref{Lemma}{lem:#1}}
\newcommand{\claimlab}[1]{\label{claim:#1}}
\newcommand{\claimref}[1]{\namedref{Claim}{claim:#1}}
\newcommand{\corlab}[1]{\label{cor:#1}}
\newcommand{\corref}[1]{\namedref{Corollary}{cor:#1}}

\newcommand{\applab}[1]{\label{app:#1}}
\newcommand{\appref}[1]{\namedref{Appendix}{app:#1}}

\newcommand{\figlab}[1]{\label{fig:#1}}
\newcommand{\figref}[1]{\namedref{Figure}{fig:#1}}

\newcommand{\deflab}[1]{\label{def:#1}}

\newcommand{\exlab}[1]{\label{ex:#1}}
\newcommand{\exref}[1]{\namedref{Example}{ex:#1}}

\newenvironment{remindertheorem}[1]{\medskip \noindent {\textcolor{darkgray}{$\blacktriangleright$}\nobreakspace\sffamily\bfseries Reminder of  #1.  }\em}{}
\def \parents  {\mdef{\mathsf{parents}}}
\def \pROM			{\mdef{\mathsf{pROM}}}
\def \children  {\mdef{\mathsf{children}}}
\def \depth  {\mdef{\mathsf{depth}}}
\def \indeg  {\mdef{\mathsf{indeg}}}
\def \outdeg  {\mdef{\mathsf{outdeg}}}
\def \Infl	 {\mdef{\mathsf{Infl}}}
\def \var    {\mdef{\mathsf{Var}}}
\def \idr  {\mdef{\mathsf{IDR}}}
\def \bitsparsify  {\mdef{\mathsf{Sparsify}}}

\def \superconc  {\mdef{\mathsf{superconc}}}
\def \cc       {\mdef{\mathsf{cc}}}       
\def \inp {\mdef{\mathsf{input}}}
\def \outp {\mdef{\mathsf{output}}}
\def \inte {\mdef{\mathsf{interior}}}

\newcommand{\Peb}{{\cal P}} 

\newcommand{\pPeb}{\Peb^{\parallel}}

\def \pcc {\cc} 

\renewcommand{\O}[1]{\ensuremath{\mathcal{O}\left(#1\right)}}
\newcommand{\eps}{\varepsilon}

\newcommand{\mdef}[1]{{\ensuremath{#1}}\xspace}  

\newcommand{\superscript}[1]{\ensuremath{^{\mbox{\tiny{\textit{#1}}}}}\xspace}
\def \th {\superscript{th}}     
\def \nd {\superscript{nd}}     
\def \etal{{\it et~al.}}

\newcommand{\ignore}[1]{}

\newif\ifnotes\notestrue 
\ifnotes
\newcommand{\samson}[1]{\textcolor{purple}{{\bf (Samson:} {#1}{\bf ) }} \marginpar{\tiny\bf
             \begin{minipage}[t]{0.5in}
               \raggedright S:
            \end{minipage}}}  
\newcommand{\jeremiah}[1]{\textcolor{red}{{\bf (Jeremiah:} {#1}{\bf ) }} \marginpar{\tiny\bf
             \begin{minipage}[t]{0.5in}
               \raggedright J:
            \end{minipage}}}
    
\else
\newcommand{\samson}[1]{}
\newcommand{\jeremiah}[1]{}
\fi

\hypersetup{
     colorlinks = true,
     citecolor = darkpastelgreen,
     linkcolor = red,
     urlcolor = black
}

\makeatletter
\renewcommand*{\@fnsymbol}[1]{\textcolor{blue}{\ensuremath{\ifcase#1\or *\or \dagger\or \ddagger\or
 \mathsection\or \triangledown\or \mathparagraph\or \|\or **\or \dagger\dagger
   \or \ddagger\ddagger \else\@ctrerr\fi}}}
\makeatother

\providecommand{\email}[1]{\href{mailto:#1}{\nolinkurl{#1}\xspace}}

\definecolor{mahogany}{rgb}{0.75, 0.25, 0.0}
\definecolor{darkblue}{rgb}{0.0, 0.0, 0.55}
\definecolor{darkpastelgreen}{rgb}{0.01, 0.75, 0.24}
\definecolor{darkgreen}{rgb}{0.0, 0.2, 0.13}
\definecolor{darkgoldenrod}{rgb}{0.72, 0.53, 0.04}
\definecolor{forestgreen}{rgb}{0.13, 0.55, 0.13}
\definecolor{darkred}{rgb}{0.55, 0.0, 0.0}
\definecolor{purduedigitalheadlinegold}{HTML}{98700D}
\definecolor{purduecampusgold}{HTML}{C28E0E}
\definecolor{purduecoalgray}{HTML}{4D4038}
\definecolor{purdueevertrueblue}{HTML}{5B6870}
\definecolor{purduemoondustgray}{HTML}{BAA892}
\definecolor{purdueslayterskyblue}{HTML}{6E99B4}
\definecolor{tangocolordarkchameleon}{HTML}{4E9A06}

\nolinenumbers
\hideLIPIcs

\title{Approximating Cumulative Pebbling Cost is Unique Games Hard} 

\titlerunning{Approximating Cumulative Pebbling Cost is Unique Games Hard}

\author{Jeremiah Blocki}{Department of Computer Science, Purdue University, West Lafayette, IN, USA \and \url{https://www.cs.purdue.edu/homes/jblocki} }{jblocki@purdue.edu}{https://orcid.org/0000-0002-5542-4674}{Research supported in part by NSF Award \#1755708.}

\author{Seunghoon Lee}{Department of Computer Science, Purdue University, West Lafayette, IN, USA \and \url{https://www.cs.purdue.edu/homes/lee2856} }{lee2856@purdue.edu}{https://orcid.org/0000-0003-4475-5686}{Research supported in part by NSF Award \#1755708 and by the Center for Science of Information at Purdue University (NSF CCF-0939370).}

\author{Samson Zhou}{School of Computer Science, Carnegie Mellon University, Pittsburgh, PA, USA \and \url{https://samsonzhou.github.io/} }{samsonzhou@gmail.com}{https://orcid.org/0000-0001-8288-5698}{}

\authorrunning{J. Blocki and S. Lee and S. Zhou}

\Copyright{Jeremiah Blocki and Seunghoon Lee and Samson Zhou}

\ccsdesc[500]{Theory of computation~Computational complexity and cryptography}
\ccsdesc[500]{Security and privacy~Hash functions and message authentication codes}

\keywords{Cumulative Pebbling Cost, Approximation Algorithm, Unique Games Conjecture, $\gamma$-Extreme Depth Robust Graph, Superconcentrator, Memory-Hard Function}

\category{}

\supplement{}

\funding{The opinions in this paper are those of the authors and do not necessarily reflect the position of the National Science Foundation.}

\acknowledgements{Part of this work was done while Samson Zhou was a postdoctoral fellow at Indiana University.}

\EventEditors{Thomas Vidick}
\EventNoEds{1}
\EventLongTitle{11th Innovations in Theoretical Computer Science Conference (ITCS 2020)}
\EventShortTitle{ITCS 2020}
\EventAcronym{ITCS}
\EventYear{2020}
\EventDate{January 12--14, 2020}
\EventLocation{Seattle, Washington, USA}
\EventLogo{}
\SeriesVolume{151}
\ArticleNo{13}

\begin{document}
\maketitle
\begin{abstract}
The cumulative pebbling complexity of a directed acyclic graph $G$ is defined as $\pcc(G) = \min_P \sum_i |P_i|$, where the minimum is taken over all legal (parallel) black pebblings of $G$ and $|P_i|$ denotes the number of pebbles on the graph during round $i$. 
Intuitively, $\pcc(G)$ captures the amortized Space-Time complexity of pebbling $m$ copies of $G$ in parallel. 
The cumulative pebbling complexity of a graph $G$ is of particular interest in the field of cryptography as $\pcc(G)$ is tightly related to the amortized Area-Time complexity of the Data-Independent Memory-Hard Function (iMHF) $f_{G,H}$~\cite{STOC:AlwSer15} defined using a constant indegree directed acyclic graph (DAG) $G$ and a random oracle $H(\cdot)$. 
A secure iMHF should have amortized Space-Time complexity as high as possible, e.g., to deter brute-force password attacker who wants to find $x$ such that $f_{G,H}(x) = h$. 
Thus, to analyze the (in)security of a candidate iMHF $f_{G,H}$, it is crucial to estimate the value $\pcc(G)$ but currently, upper and lower bounds for leading iMHF candidates differ by several orders of magnitude. 
Blocki and Zhou recently showed that it is \textsf{NP}-Hard to compute $\pcc(G)$, but their techniques do not even rule out an efficient $(1+\eps)$-approximation algorithm for any constant $\eps>0$. 
We show that for \emph{any} constant $c > 0$, it is Unique Games hard to approximate $\pcc(G)$ to within a factor of $c$. 

Along the way, we show the hardness of approximation of the DAG Vertex Deletion problem on DAGs of constant indegree. 
Namely, we show that for any $k,\eps >0$ and given a DAG $G$ with $N$ nodes and constant indegree, it is Unique Games hard to distinguish between the case that $G$ is $(e_1, d_1)$-reducible with $e_1=N^{1/(1+2\eps)}/k$ and $d_1=k N^{2\eps/(1+2\eps)}$, and the case that $G$ is $(e_2, d_2)$-depth-robust with $e_2 = (1-\eps)k e_1$ and $d_2= 0.9 N^{(1+\eps)/(1+2\eps)}$, which may be of independent interest. 
Our result generalizes a result of Svensson who proved an analogous result for DAGs with indegree $\O{N}$. 
\end{abstract}


\section{Introduction}
The black pebbling game is a powerful abstraction that allows us to analyze the complexity of functions $f_G$ with a static data-dependency graph $G$. 
In particular, a directed acyclic graph (DAG) $G=(V,E)$ can be used to encode data-dependencies between intermediate values produced during computation e.g., if $L_v$ is the $v\th$ intermediate value and $L_v:= L_{j} \times L_i$ then the DAG $G$ would include directed edges $(i,v)$ and $(j,v)$ indicating that $L_v$ depends on the previously computed values $L_i$ and $L_j$. 
A black pebbling of $G$ is a sequence $P=(P_0,\ldots, P_t) \subseteq V$  of pebbling configurations. 
Intuitively, a pebbling configuration $P_i$ describes the set of data labels that have been computed and stored in memory at time $i$. 
The rules of the pebbling game stipulate that we must have $\parents(v) = \{ u ~: ~(u,v) \in E\} \subseteq P_i$ for each newly pebbled node $v \in P_{i+1}\setminus P_i$ i.e., before we can compute a new data value $L_v$, we must first have the labels of each dependent data value $L_u$ available in memory.

Historically, much of the literature has focused on the sequential black pebbling game where we require that $|P_{i+1}\setminus P_i| \leq 1$ for all round $i$. 
In recent years, the parallel black pebbling game has seen renewed interest due to the rapid expansion of parallel computing, e.g., GPUs, FPGAs. 
In the more general parallel black pebbling game, there is no such restriction on the number of new pebbles in each round, i.e., on a parallel architecture, it is possible to determine $L_v$ for each node $v \in P_{i+1}\setminus P_i$ simultaneously since the dependent data-values are already in memory.


There are several natural ways to measure the cost of a pebbling. 
The space complexity of a DAG $G$ asks for a legal pebbling $P=(P_0,\ldots, P_t)$ that minimizes the maximum space usage $\max_{i \leq t} |P_i|$ --- even if the time $t$ is exponential in the number of nodes $N$. 
Space-time complexity asks for a legal pebbling $P=(P_0,\ldots, P_t)$ that minimizes the space-time product $t \times \max_{i \leq t} |P_i|$. 
Alwen and Serbinenko~\cite{STOC:AlwSer15} observed that in the parallel black pebbling game, the space-time of pebbling $G^{\times m}$, $m$ independent copies of a DAG $G$, does not always scale linearly with $m$. 
In particular, for some DAGs $G$ the total space-time cost of pebbling $G^{\times m}$ is roughly equal to the space-time cost of pebbling a single instance of $G$ for $m = \tilde{\mathcal{O}}(\sqrt{N})$! 

Alwen and Serbinenko~\cite{STOC:AlwSer15} introduced the notion of the cumulative pebbling cost $\pcc(G)$ of a DAG $G$ to model the amortized space-time costs in the parallel black pebbling game. 
Formally, the cumulative pebbling cost of a pebbling $P$ is given by $\pcc(P) = \sum_i |P_i|$ and $\pcc(G) = \min_{P} \pcc(P)$, where the minimum is taken over all legal (parallel) black pebblings of $G$. 
The cumulative pebbling cost is a fundamental metric that is worth studying. 
It captures the amortized space-time cost of pebbling $m$ copies of $G$ in parallel, i.e., in the limit we have $\pcc(G) = \lim_{m \rightarrow \infty} \mathtt{ST}\left(G^{\times m} \right)/m$ where the space-time cost of a pebbling $P=(P_1,\ldots,P_t)$ is $\mathtt{ST}(P) = t \times \max_{i} |P_i|$ and the notation $G^{\times m}$ denotes a new graph consisting of $m$ disjoint copies of $G$. 

In this paper, we address the following question: 
\begin{quote}
\emph{Given a DAG $G$, can we (approximately) compute $\pcc(G)$?}
\end{quote}
This is a natural question in settings where we want to evaluate the function $f_G$ (with data-dependency DAG $G$) on many distinct inputs --- $\pcc(G)$ models the amortized cost of computing $f_G$. The question is also highly relevant to the cryptanalysis of Data-Independent Memory-Hard Functions (iMHFs). 
In the context of password hashing we want to find a (constant indegree) DAG $G$ with maximum cumulative pebbling complexity, e.g., to maximize the cost of a brute-force attacker who wants to evaluate the function $f_G$ on every input in a password cracking dictionary. 
Thus, given a DAG $G$ one might wish to lower-bound $\pcc(G)$ before using $G$ in the design of a memory-hard password hashing algorithm.


\paragraph*{Cumulative Pebbling Complexity in Cryptography.} 
In many natural contexts such as password hashing and Proofs of Work, it is desirable to lower bound the amortized space-time cost, e.g., in the random oracle model it is known that the cumulative memory complexity of a (side-channel resistant) iMHF $f_{G,H}$ is $\Omega(\pcc(G))$, where $f_{G,H}$ is a labeling function defined in terms of the DAG $G$ and a random oracle $H$~\cite{STOC:AlwSer15}. 
Thus, in the field of cryptography there has been a lot of interest in designing constant indegree graphs with cumulative pebbling cost $\pcc(G)$ as large as possible and in analyzing the pebbling cost $\pcc(G)$ of candidate iMHF constructions $f_{G,H}$, e.g., see \cite{C:AlwBlo16,EC:AlwBloPie17,ESP:AlwBlo17,EC:AlwBloPie18,CCS:AlwBloHar17,TCC:BloZho17}. 

From an asymptotic standpoint many of the open questions have been (nearly) resolved. 
Alwen and Blocki~\cite{C:AlwBlo16} showed that for any DAG $G$ with $N$ nodes and constant indegree we have $\pcc(G) = \O{N^2 \log \log N/\log N}$, while Alwen \etal~\cite{EC:AlwBloPie17,CCS:AlwBloHar17} gave constructions with $\pcc(G) = \Omega(N^2/\log N)$. 
For Argon2i, the winner of the password hashing competition, we have the upper bound $\pcc(G) = \O{N^{1.767}}$ and the lower bound $\pcc(G) = \tilde{\Omega}\left( N^{1.75}\right)$~\cite{TCC:BloZho17}. 

Most of these upper/lower bounds exploited a relationship between $\pcc(G)$ and a combinatorial property called depth-robustness. 
A DAG $G=(V,E)$ is $(e,d)$-reducible if we can find a subset $S \subseteq V$ with $|S|\leq e$ such that any directed path $P$ in $G$ of length $d$ contains at least one node in $S$. 
On the other hand, if $G$ is not $(e,d)$-reducible, then we say that $G$ is $(e,d)$-depth robust. 
Depth-robustness is known to be both necessary~\cite{C:AlwBlo16} and sufficient~\cite{EC:AlwBloPie17} for secure iMHFs. 
In particular, any $(e,d)$-reducible DAG $G$ with $N$ nodes and indegree $\indeg(G)$ has $\pcc(G) \leq \min_{g \geq d} \left( eN + gN \times \indeg(G) + \frac{N^2 d}{g} \right)$~\cite{C:AlwBlo16} while any $(e,d)$-depth robust DAG $G$ has $\pcc(G) \geq ed$~\cite{EC:AlwBloPie17}. 
The later observation was used to build a constant indegree graph $G$ with $\pcc(G) = \Omega(N^2/\log N)$ by showing that the constructed $G$ is $\left( \Omega(N/\log N), \Omega(N)\right)$-depth robust. 
The former observation was used to prove that any constant indegree graph has  $\pcc(G) = \O{N^2 \log \log N/\log N}$ by exploiting the observation that any such DAG $G$ is $\left( \O{N \log \log N/\log N}, \Omega(N/\log^2 N)\right)$-reducible (simply set $g= \O{N \log \log N/\log N}$ in the above \cite{C:AlwBlo16} bound). 

Although many of the open questions have been (nearly) resolved from an asymptotic standpoint, from a concrete security standpoint for all practical iMHF candidates $G$, the best known upper and lower bounds on $\pcc(G)$ differ by several orders of magnitude. 
In fact, Blocki \etal~\cite{EPRINT:BHKLXZ18} recently found that for practical parameter settings ($N \leq 2^{24}$), Argon2i provides better resistance to known pebbling attacks than DRSample~\cite{CCS:AlwBloHar17} despite the fact that DRSample ($\pcc(G) = \Omega(N^2/\log N)$) is asymptotically superior to Argon2i ($\pcc(G) = \tilde{\Omega}\left( N^{1.75}\right)$). 
Of course it is certainly possible that an improved pebbling strategy for Argon2i will reverse this finding tomorrow making it difficult to provide definitive recommendations about which construction is superior in practice.

Given a DAG $G$, one might try to resolve these questions directly by (approximately) computing $\pcc(G)$. 
Blocki and Zhou~\cite{FC:BloZho18} previously showed that the problem of computing $\pcc(G)$ is {\sf NP}-Hard. 
However, their result does not even rule out the existence of a $(1+\eps)$-approximation algorithm for any constant $\eps >0$. 

\subsection{Our Contributions}
Our main result is the hardness of any constant factor approximation to the cost of graph pebbling even for DAGs with constant indegree\footnote{Each node $v$ in a data-dependency DAG $G$ model an atomic unit of computation. Thus, in practice we expect $G$ to have indegree $2$ or $3$. If $L_v = g(L_{v_1},\ldots,L_{v_k})$ is a function of $k \gg 2$ previously computed values $L_{v_1},\ldots, L_{v_k}$ then we would have generated several additional intermediate data-values while evaluating $g(\cdot)$. These data-values should have been included as nodes in $G$ which is supposed to have a node for every intermediate data-value.}.
\begin{theorem}
Given a DAG $G$ with constant indegree, it is Unique Games hard to approximate $\pcc(G)$ within any constant factor. (See \thmref{thm:approx}.)
\end{theorem}

Along the way to proving our main result, we show that for any constant $k >0, \eps >0$, given a constant indegree graph $G$, it is Unique Games hard to distinguish between the following two cases: (1) $G$ is $(e_1, d_1)$-reducible with $e_1=N^{1/(1+2\eps)}/k$ and $d_1=k N^{2\eps/(1+2\eps)}$ and (2) $G$ is $(e_2, d_2)$-depth-robust with $e_2 = (1-\eps)k e_1$ and $d_2= 0.9 N^{(1+\eps)/(1+2\eps)}$. 
This intermediate result (see \corref{cor:second:idr})  generalizes a result of Svensson~\cite{Svensson12}, who proved an analogous result for DAGs $G$ with arbitrarily large indegree $\indeg(G)=\O{N}$. 

\corref{cor:second:idr} may be of independent interest as depth-robust graphs have found many other applications in cryptography including Proofs of Sequential Work~\cite{ITCS:MahMorVad13}, Proofs of Space~\cite{C:DFKP15}, Proofs of Replication~\cite{ITCS:Pietrzak19a,EC:Fisch19} and (relaxed) locally correctable codes for computationally bounded channels~\cite{blocki2019relaxed,BlockiKZ19}. 
Testing the depth-robustness of a DAG $G$ is especially relevant to the analysis of (tight) Proofs of Space/Replication --- several  constructions rely on (unproven) conjectures about the concrete depth-robustness of particular DAGs e.g., see \cite{EPRINT:CecMieJue18,EC:Fisch19}.

\subsection{Technical Ingredients}
To prove our result we use three technical ingredients. 
The first ingredient is a reduction of Svensson \cite{Svensson12} that it is Unique Games hard to distinguish between a DAG $G$ (with $\indeg(G) =\O{N}$) that is $(e_1, d_1)$-reducible or $(e_2, d_2)$-depth-robust. 
The second  technical ingredient is $\gamma$-Extreme Depth-Robust Graphs~\cite{EC:AlwBloPie18} with bounded indegree. 
We use $\gamma$-Extreme Depth-Robust Graphs to modify the construction of Svensson \cite{Svensson12} and show that the same result holds for graphs with much smaller indegree. 
Finally, we use low depth superconcentrators to boost the lower bound on $\pcc$ to $\min\{ e_2N, d_2N\}/8$ instead of $e_2d_2$ in the case the graph is $(e_2,d_2)$-depth robust. 
 We prove that this can be done without significantly increasing the pebbling cost in the case the graph is $(e_1,d_1)$-reducible. 

\subsubsection{Technical Ingredient 1} Our first technical ingredient is a result of Svensson \cite{Svensson12}, who proved that for any constant $k > 0, \eps >0$, it is Unique Games hard to distinguish between the following two cases (1) $G$ is $(e_1,d_1)$-reducible with $e_1 = N/k$ and $d_1 = k$, or (2) $G$ is $(e_2,d_2)$-depth robust with $e_2 = N(1-1/k)$ and $d_2 = \Omega(N^{1-\eps})$. 
To prove this, Svensson gave a reduction that transforms from any instance of Unique Games $\mathcal{U}$ to a directed acyclic graph $G_{\mathcal{U}}$ on $N$ nodes such that $G_{\mathcal{U}}$ is $(e_1,d_1)$-reducible for $e_1 \approx N/k$ and $d=k$ if $\mathcal{U}$ is satisfiable. 
Otherwise, if $\mathcal{U}$ is unsatisfiable, it can be shown that $G_{\mathcal{U}}$ is $(e_2,d_2)$-depth robust. 
This is a potentially useful starting point because the pebbling complexity of a graph $G_{\mathcal{U}}$ is closely related to its depth-robustness. 
In particular, in the second case, a result of Alwen \etal~\cite{EC:AlwBloPie17} establishes that $\pcc(G_{\mathcal{U}}) \geq e_2 d_2$ and in the first case, a result of Alwen and Blocki shows that $\pcc(G_{\mathcal{U}}) \leq \min_{g \geq d_1} \left( e_1 N + gN \times \indeg(G_{\mathcal{U}}) + \frac{N^2 d_1}{g}\right)$ ~\cite{C:AlwBlo16}.

\paragraph*{Challenges of Applying Svensson's Construction.} 
While the pebbling complexity of $G_{\mathcal{U}}$ is related to depth-robustness, there is still a vast gap between the upper/lower bounds. 
In particular, in Svensson's construction we have $\indeg(G_{\mathcal{U}}) = \O{N}$, so the $gN \times \indeg(G_{\mathcal{U}})$ term could be as large as $gN^2 \gg e_2d_2$. 
Thus, we would need to be able to reduce the indegree significantly to obtain a gap between $\pcc(G_{\mathcal{U}}) $ in the two cases. (In fact, we can show the that pebbling cost is exactly $\pcc\left(G_{\mathcal{U}}\right)=\frac{N(L+1)}{2}$ independent of the Unique Games instance $\mathcal{U}$  --- see  \lemref{PebblingCostSvensson} in the appendix.)
We remark that a na\"{i}ve attempt to reduce indegree in Svensson's construction $G_{\mathcal{U}}$ by replacing every node $v$ (as in \cite{EC:AlwBloPie17}) with a path of length $N + \indeg(v)$ would result in a constant indegree graph $G'_{\mathcal{U}}$ with $N' \approx 2N^2$ nodes that will not be useful for our purposes. 
The new graph $G'_{\mathcal{U}}$ would be $(e_1,d_1)$-reducible in the first case with $e_1 = N/k = \O{\sqrt{N'}/k}$ and $d_1 = 2kN = \O{\sqrt{N'}k}$. 
In the second case, the DAG $G'_{\mathcal{U}}$ would be  $(e_2,d_2)$-depth robust with $e_2 \approx k e_1$ and $d_2 = \O{N'^{1-\eps/2}}$. 
We would now have $\pcc(G'_{\mathcal{U}}) \leq  \min_{g \geq d_1} \left( e_1 N' + 2gN'  + \frac{N'^2 d_1}{g}\right) = \omega(e_1 N')$ for our upper bound while the lower bound is at most $e_2 d_2 \approx ke_1 N'^{1-\eps/2}$. 
At the end of the day, the graph $G_{\mathcal{U}}$ is still quite far from what we need. 

\subsubsection{Technical Ingredient 2: $\gamma$-Extreme Depth-Robust Graphs.}
It does not seem to be possible to obtain a suitable graph $G_{\mathcal{U}}$ by applying indegree reduction techniques to Svensson's Construction in a black-box manner. 
Instead, we open up the black-box and show how to reduce the indegree using a recent technical result of Alwen \etal~\cite{EC:AlwBloPie18}. 
A DAG $G_{\gamma,N}$ on $N$ nodes is said to be $\gamma$-extreme depth-robust if it is $(e,d)$-depth robust for any $e,d > 0$ such that $e+d \leq (1-\gamma) N$. 
Alwen \etal~\cite{EC:AlwBloPie18} showed that for any constant $\gamma > 0$, there exists a family $\{G_{\gamma,N}\}_{N=1}^\infty$ of $\gamma$-extreme depth robust DAGs with maximum indegree $\O{\log N}$. 
While Alwen \etal~\cite{EC:AlwBloPie18} were not focused on outdegree, it is not too difficult to see that their construction yields a single family of DAGs with maximum indegree and outdegree $\O{\log N}$. 

In Svensson's construction, the DAG $G_{\mathcal{U}}$ is partitioned into $L = N^{1-\eps}$ symmetric layers i.e., if $u_{\ell}$ (the copy of node $u$ in layer $\ell_1$) is connected to $v_{\ell_2}$ (the copy of node $v$ in layer $\ell_2 > \ell_1$) then for {\em any} layers $i < j \leq L$, the directed edge $(u_i,v_j)$ exists. 
The fact that this edge is ``copied'' $\O{L^2}$ times for every pair of layers $i < j$ significantly increases the indegree. 
However, Svensson's argument that $G_{\mathcal{U}}$ is depth-robust in the second case relies on the existence of each of these edges. 
To reduce the indegree we start with a $\gamma$-extreme depth robust DAG $G_{\gamma,L}$ on $L$ nodes and only keep edges between nodes $u_i$ and $v_j$ in layers $i$ and $j$ if there is a path of length $\leq 2$ between nodes $i$ and $j$ in $G_{\gamma,L}$. 
The new graph can also be shown to have degree at most $\O{\indeg(G_L) \times \outdeg(G_L) \times N/L} = \O{N^\eps \log^2 N}$. 
Despite the fact that the indegree is vastly reduced, we are still able to modify Svensson's argument to prove that (for a suitable constant $\gamma > 0$) our new graph is still $(e_2,d_2)$-depth robust with $e_2 \approx k e_1$ and $d_2 = \O{N^{1-\eps}}$ --- note that the new graph is clearly still $(e_1,d_1)$-reducible if $\mathcal{U}$ is satisfiable since we only remove edges from Svensson's construction.

We can then apply the generic black-box indegree reduction of \cite{EC:AlwBloPie17} to reduce the indegree to $2$ by replacing every node with a path of length $N^{2\eps}$. 
This established our first technical result that even for constant indegree DAGs, it is Unique Games hard to distinguish between the following two cases: (1) $G$ is $(e_1, d_1)$-reducible with $e_1=N^{1/(1+2\eps)}/k$ and $d_1=k N^{2\eps/(1+2\eps)}$, and (2) $G$ is $(e_2, d_2)$-depth-robust with $e_2 = (1-\eps)k e_1$ and $d_2= 0.9 N^{(1+\eps)/(1+2\eps)}$.

\subsubsection{Technical Ingredient 3: Superconcentrators} 
Although indegree reduction is a crucial step toward showing hardness of approximation for graph pebbling complexity, we still cannot apply known results that relate $(e_1,d_1)$-reducibility and $(e_2,d_2)$-depth robustness to pebbling complexity, since there is still no gap between the pebbling complexity of the two cases. 
In particular, we are always stuck with the $e_1N$ term in the upper bound of \cite{C:AlwBlo16} which is {\em already} much larger than the lower bound $e_2d_2$ from \cite{EC:AlwBloPie18}. 
To overcome this result we rely on superconcentrators. 
A \emph{superconcentrator} is a graph that connects $N$ input nodes to $N$ output nodes so that any subset of $k$ inputs and $k$ outputs are connected by $k$ vertex disjoint paths. 
Moreover, the total number of edges in the graph should be $\O{N}$. 

Blocki \etal\,\cite{EPRINT:BHKLXZ18} recently proved that $G'$, the \emph{superconcentrator overlay} of an $(e,d)$-depth robust graph, has pebbling cost $\pcc(G') \geq \max\{ eN, dN\}/8$, which is a significant improvement on the lower bound $\pcc(G') \geq ed$ when $e = o(N)$ and $d=o(N)$. 
This allows us to increase the lower-bound in case 2, but we need to be careful that we do not significantly increase the pebbling cost in case 1. 
To do this we rely on the existence of superconcentrators with depth $\O{\log N}$~\cite{Pippenger77} and we give a significantly improved pebbling attack on the superconcentrator overlay DAG $G'$ in case 1 when the original graph is $(e_1,d_1)$-reducible. 
With the improved pebbling attack, we are able to show that $\pcc(G) \geq e_1kN/16$ in case 2 and that $\pcc(G) \leq 16e_1N$ in case 1. 
Since $k$ is an arbitrary constant, this implies that it is Unique Games hard to approximate $\pcc(G)$ to within any constant factor $c >0$.

\section{Related Work}
Pebbling games have found a number of applications under various formulations and models (see the survey \cite{Nordstrom13} for a more thorough review). 
The sequential black pebbling game was introduced by Hewitt and Paterson~\cite{HP70}, and by Cook~\cite{Coo73} and has been particularly useful in exploring space/time trade-offs for various problems like matrix multiplication~\cite{Tom78}, fast fourier transformations~\cite{SS78,Tom78}, integer multiplication~\cite{SS79a} and many others~\cite{Cha73,SS79b}. 
In cryptography it has been used to construct/analyze Proofs of Space \cite{C:DFKP15,RenD16}, Proofs of Work~\cite{C:DwoNaoWee05,ITCS:MahMorVad13} and Memory-Hard Functions~\cite{AC:ForLucWen14}. Alwen and Serbinenko~\cite{STOC:AlwSer15} argued that the parallel version of the black pebbling game was more appropriate for Memory-Hard Functions and they proved that any iMHF attacker in the parallel random oracle model corresponds to a pebbling strategy with equivalent cumulative memory cost.

The space cost of the black pebbling game is defined to be $\max_{i}|P_i|$, which intuitively corresponds to minimizing the maximum space required during computation of the associated function.   
Gilbert \etal~\cite{GilbertLT79} studied the space-complexity of the black-pebbling game and showed that this problem is \textsf{PSPACE}-Complete by reducing from the truly quantified boolean formula (TQBF) problem. In our case, the decision problem is $\pcc(G) \leq k$ is in $\mathsf{NP}$ because the optimal pebbling strategy cannot last for more than $N^2$ steps since {\em any} graph with $N$ nodes  has $\pcc(G) \leq N^2$. 

\paragraph*{Red-Blue Pebbling.} 
Given a DAG $G=(V,E)$, the goal of the red-blue pebbling game~\cite{HongK81} is to place pebbles on all sink nodes of $G$ (not necessarily simultaneously) from an empty starting configuration. Intuitively, red pebbles represent values in cache and blue pebbles represent values stored in memory. Blue pebbles must be converted to red pebbles (e.g., loaded into cache) before they can be used in computation, but there is a limit $m$ (cache-size) on the number of red-pebbles that can be used. Red-blue pebbling games have been used to study memory-bound functions~\cite{DGN03} (functions that incur many expensive cache-misses~\cite{NDSS:AbadiBW03}). 

Ren and Devadas introduced the notion of bandwidth hard functions and used the red-blue pebbling game to analyze the energy cost of a memory hard function~\cite{TCC:RenDev17}. 
In their model, red-moves (representing computation performed using data in cache) have a smaller cost $c_r$ than blue-moves $c_b$ (representing data movements to/from memory) and a DAG $G$ on $N$ nodes is said to be bandwidth hard if any red-blue pebbling has cost $\Omega(N \cdot c_b)$.  
Ren and Devadas showed that the bit reversal graph~\cite{LT82}, which forms the core of iMHF candidate Catena-BRG~\cite{AC:ForLucWen14}, is maximally bandwidth hard. 
Subsequently, Blocki \etal~\cite{CCS:BloRenZho18} gave a pebbling reduction showing that any attacker random oracle model (\pROM) can indeed be viewed as a red-blue pebbling with equivalent cost. 
They also show that it is \textsf{NP}-Hard to compute the minimum cost red-blue pebbling of a DAG $G$ i.e., the decision problem ``is the red-blue pebbling cost $\leq k$?'' is  \textsf{NP}-Complete (A result of Demaine and Liu~\cite{DemaineL17,Liu17} implies that the problem is \textsf{PSPACE}-Hard to compute the red-blue pebbling cost when $c_r=0$ i.e., computation is free). In general, the red-blue cost of $G$ is always lower bounded by $c_r N$ and upper-bounded by $2c_b N + c_rN$. The question of a more efficient $c$-approximation algorithm for $c=o(c_b/c_r)$ remains open.  



\paragraph*{Unique Games.} 
Recently, the Unique Games Conjecture and related conjectures have received a lot of attention for their applications in proving hardness of approximation. 
Khot \etal~\cite{KhotKMO07} showed that the Goemans-Williamson approximation algorithm for Max-Cut~\cite{STOC:GoeWil94} is optimal, assuming the Unique Games Conjecture. 
Khot and Regev~\cite{KhotR08} showed that Minimum Vertex Cover problem is Unique Games hard to solve within a factor of $2-\eps$, which is nearly tight from the guarantee that a simple greedy algorithm gives. 
The Unique Games Conjecture also leads to tighter approximation hardness for other problems including Max 2-SAT~\cite{KhotKMO07} and Betweenness~\cite{CharikarGM09}. 
Although a previous stronger version of the conjecture asked whether Unique Games instances required exponential time algorithms in the worst case, Arora \etal~\cite{FOCS:AroBarSte10} gave a subexponential time algorithm for Unique Games. 
Lately, focus has also been drawn toward studying the related Label Cover Problem, such as the $2$-Prover-$1$-Round Games, i.e. the 2-to-1 Games Conjecture \cite{DinurKKMS18} and the 2-to-2 Games Conjecture \cite{KhotMS17}. 

\section{Preliminaries}
We use the notation $[N]$ to denote the set $\{0,1,\ldots,N-1\}$. 
Given a directed acyclic graph $G=(V,E)$ and a node $v \in V$, we use $\parents(v)= \{u~:~(u,v) \in E\}$ (resp. $\children(v) = \{ u~:~(v,u) \in E\}$) to denote the parents (resp. children) of node $v$. 
We use $\indeg(v) = \left| \parents(v)\right|$ (resp. $\outdeg(v) = \left|\children(v)\right|$) to denote the number of incoming (resp. outgoing) edges into (resp. out of) the vertex $v$. 
We also define $\indeg(G)=\underset{v\in V}{\max}\,\indeg(v)$ and $\outdeg(G)=\underset{v\in V}{\max}\,\outdeg(v)$. 
Given a set $S \subseteq V$ of nodes, we use $G-S$ to refer to the graph obtained by deleting all nodes in $S$ and all edges incident to $S$. 
We also use $G[S] = G-(V\setminus S)$ to refer to the subgraph induced by the nodes $S$, i.e., deleting every other node in $V\setminus S$. 
Given a node $v \not \in S$, we use $\depth(v,G-S)$ to refer to the longest directed path in $G-S$ ending at node $v$ and we use $\depth(G-S) = \max_{v \not \in S} \depth(v,G-S)$ to refer to the longest directed path in $G-S$. 
Given a subset $B$, we will also use $\depth_B(v,G-S)$ to refer to the maximum number of nodes in the set $B$ contained in any directed path in $G-S$ that ends at node $v$. 
We define $\depth_B(G-S)= \max_{v \not \in S} \depth_B(v,G-S)$ analogously. 

\begin{definition}[Unique Games]
An instance $\mathcal{U}=(G=(V,W,E),[R],\{\pi_{v,w}\}_{v,w})$ of Unique Games consists of a regular bipartite graph $G(V,W,E)$ and a set $[R]$ of labels. 
Each edge $(v,w)\in E$ has a constraint given by a permutation $\pi_{v,w}:[R]\rightarrow[R]$. 
The goal is to output a labeling $\rho:(V\cup W)\rightarrow[R]$ that maximizes the number of satisfied edges, where an edge is satisfied if $\rho(v)=\pi_{v,w}(\rho(w))$.
\end{definition}

\begin{conjecture}[Unique Games Conjecture]
\cite{STOC:Khot02b}
For any constants $\alpha,\beta>0$, there exists a sufficiently large integer $R$ (as a function of $\alpha,\beta$) such that for Unique Games instances with label set $[R]$, no polynomial time algorithm can distinguish whether: (1) the maximum fraction of satisfied edges of any labeling is at least $1-\alpha$, or (2) the maximum fraction of satisfied edges of any labeling is less than $\beta$. 
\end{conjecture}

\paragraph*{Graph Pebbling.} 
The goal of the (black) pebbling game is to place pebbles on all sink nodes of some input directed acyclic graph (DAG) $G=(V,E)$. 
The game proceeds in rounds, and each round $i$ consists of a number of pebbles $P_i\subseteq V$ placed on a subset of the vertices. 
Initially, the graph is unpebbled, $P_0=\emptyset$, and in each round $i\ge 1$, we may place a pebble on $v\in P_i$ if either all parents of $v$ contained pebbles in the previous round ($\parents(v)\subseteq P_{i-1}$) or if $v$ already contained a pebble in the previous round ($v\in P_{i-1}$). 
In the sequential pebbling game, at most one new pebble can be placed on the graph in any round (i.e., $\left|P_i \backslash P_{i-1} \right| \leq 1)$, but this restriction does not apply in the parallel pebbling game. 

We use $\pPeb_G$ to denote the set of all valid parallel pebblings of $G$. 
The \emph{cumulative cost} of a pebbling $P=(P_1,\ldots,P_t) \in \pPeb_G$ is the quantity $\cc(P):=|P_1|+\ldots+|P_t|$ that represents the sum of the number of pebbles on the graph during every round. 
The (parallel) \emph{cumulative pebbling cost} of $G$, denoted $\pcc(G) := \min_{P \in \pPeb_G} \cc(P)$, is the cumulative cost of the best legal pebbling of $G$. 

A DAG $G$ is \emph{$(e,d)$-reducible} if there exists a subset $S\subseteq V$ of size $|S| \leq e$ such that $\depth(G-S) < d$. 
That is, there are no directed paths containing $d$ vertices remaining, once the vertices in the set $S$ are removed from $G$. 
If $G$ is not $(e,d)$-reducible, we say that it is $(e,d)$-depth robust.

\section{Reduction}
Svensson \cite{Svensson12} showed that for any constant $k,\epsilon >0$ it is Unique Games hard to distinguish between whether a DAG  $G$  is $(e_1,d_1)$-reducible for $e_1 = N/k$ and $d_1 = k$ or $G$ is $(e_2,d_2)$-depth robust with $e_2 = N(1-1/k)$ and $d_2 = \Omega(N^{1-\eps})$. 
To prove this, Svensson showed how to transform a Unique Games instance $\mathcal{U}=(G=(V,W,E),[R],\{\pi_{v,w}\}_{v,w})$ into a graph $G_{\mathcal{U}}$ such that $G_{\mathcal{U}}$ is $(e_1,d_1)$-reducible if it is possible to satisfy $1-\alpha$ fraction of the edges and $G_{\mathcal{U}}$ is $(e_2,d_2)$-depth robust if it is not possible to satisfy $\beta$-fraction of the edges. 
To obtain inapproximability results for $\pcc$, it is crucial to substantially reduce the indegree of this construction.


\subsection{Review of Svensson's Construction}
To construct $G_{\mathcal{U}}$, Svensson first constructs a layered bipartite DAG $\hat{G}_{\mathcal{U}}$, which encodes the unique games instance $\mathcal{U}$ and later transforms $\hat{G}_{\mathcal{U}}$ into the required DAG $G_{\mathcal{U}}$. 
For completeness, we provide a full description of the DAG $\hat{G}_{\mathcal{U}}$ in the appendix. 
We will focus our discussion here on the essential properties  of the DAG $\hat{G}_{\mathcal{U}}$.

The graph $\hat{G}_{\mathcal{U}}$ has a number of \emph{bit-vertices} $B$ partitioned into \emph{bit-layers} $B=B_0\cup\ldots\cup B_{L}$, where $B_i$ is the set of bit-vertices in bit-layer $i$. 
Each $B_i$ can be partitioned into sets $B_{i,w}$ for $w \in W$. 
Similarly, $\hat{G}_{\mathcal{U}}$ has a number of \emph{test-vertices} $T$ partitioned into \emph{test-layers} $T=T_0\cup\ldots\cup T_{L-1}$, where $T_i$ is the set of test-vertices in test-layer $i$. 
Outgoing edges for test-layer $T_\ell$ must be directed into a bit vertex in layer $B_{\ell'}$ with $\ell' > \ell$. 
Similarly, outgoing edges from $B_{\ell}$ must be directed into a test vertex in layer $T_{\ell'}$ with $\ell' \geq \ell$. 
Each $T_i$ can be partitioned into sets $T_{i,v}$ for $v \in V$. 
The constraints in our Unique Games instance $\mathcal{U}$ are encoded as edges between the bit vertices and test vertices. 
We use $N = |T|$ to denote the total number of test nodes and remark that the parameter $L$ is set such that $L \geq N^{1-\epsilon}$.

$\hat{G}_{\mathcal{U}}$ also displays symmetry between the layers in the sense that $B_\ell = \{b_1^\ell,\ldots, b_m^\ell\}$ and $T_\ell = \{t_1^\ell,\ldots, t_p^\ell\}$, so that the number of bit-vertices in each bit-layer is the same and the number of test-vertices in each test-layer is the same.

\paragraph*{Symmetry.} 
In Svensson's construction, we have exactly $m$ bit vertices in every layer $B_\ell = \{b_1^\ell,\ldots, b_m^\ell\}$ and exactly $p$ test vertices in every layer $T_\ell = \{t_1^\ell,\ldots, t_p^\ell\}$. 
The edges between $B_{\ell}$ and $T_{\ell}$ (resp. $T_{\ell}$ and $B_{\ell+1}$ ) encode the edge constraints in the unique games instance $\mathcal{U}$. 
Furthermore, the construction is symmetric so that directed edge $(b_i^\ell,t_j^\ell)$ exists if and only if for {\em every} $\ell'\geq \ell$ the edge $(b_i^\ell,t_j^{\ell'})$ exists. 
Thus for any $\ell' \geq \ell$, the edges between $B_{\ell}$ and $T_{\ell'}$ encode the constraints in $\mathcal{U}$. 
Similarly, the directed edge $(t_j^\ell,b_i^{\ell+1})$ exists if and only if any $\ell' > \ell$ the edge $(t_j^\ell,b_i^{\ell'})$ exists. 
We remark that this means that the indegree of the graph $\hat{G}_{\mathcal{U}}$ is at least $L$ (and can be as large as $\Omega(N)$ in general).

\paragraph*{Robustness of $\hat{G}_{\mathcal{U}}$.} 
Svensson argues that if it is possible to satisfy a $1-\alpha$ fraction of the constraints in $\mathcal{U}$, then there exists a subset $S \subseteq T$ of at most $|S| \leq e_1$ test-vertices such that  $\depth_B(\hat{G}_{\mathcal{U}}-S)  \leq d_1$. 
Similarly, if it is not possible to satisfy a $\beta$-fraction of the constraints, then for any subset $S \subseteq T$ of at most $|S| \leq e_2$ test-vertices, we have $\depth_B(\hat{G}_{\mathcal{U}}-S)  \geq d_2$. 
This does not directly show that $\hat{G}_{\mathcal{U}}$ is depth-robust since we are not allowed to delete bit-vertices. 
However, one can easily transform $\hat{G}_{\mathcal{U}}$ into a graph $G_{\mathcal{U}}$ on the $N = |T|$ test nodes such that $G_{\mathcal{U}}$ is $(e,d)$-depth robust if and only if for all subsets $S \subseteq T$ of $|S| \leq e$ test vertices in $\hat{G}_{\mathcal{U}}$, we have $\depth_B(\hat{G}_{\mathcal{U}}-S)  \geq d$. 
It is worth mentioning that we can view these guarantees as a form of \emph{weighted} depth-robustness where all test-vertices have weight $1$ and all bit-vertices have weight $\infty$, i.e., if $1-\alpha$ fraction of the constraints in $\mathcal{U}$, then we can find a subset $S$ of nodes with weight $\mathsf{weight}(S) \leq e_1$ such that $\depth(\hat{G}_{\mathcal{U}}-S)  \leq d_1$, and if it is not possible to satisfy $\beta$-fraction of the constraints, then for any subset $S$ with  $\mathsf{weight}(S) \leq e_2$ we have $\depth(\hat{G}_{\mathcal{U}}-S)  \geq d_1$. 

\paragraph*{Graph Coloring and Robustness.} 
An equivalent way to view the problem of {\em weighted} reducibility (resp. depth-robustness) is in terms of graph coloring. 
This view is central to Svensson's argument. 
In particular, if we can find a depth reducing set $S \subseteq T$ of size $|S| \leq e$ such that $\depth_B(\hat{G}_{\mathcal{U}}-S) \leq d$, then we can define a $d$-coloring $\chi:B \rightarrow [d]$ of each of the bit-vertices such that the coloring $\chi$ is consistent with every remaining test node $v \in T\setminus S$. 
Here, consistency means that $\max_{b \in \parents(v)} \chi(b) <  \min_{b \in \children(v)} \chi(b)$. 
In fact, it is not too difficult to see that there is a subset $S \subseteq T$ of $|S| \leq e$ test-vertices such that  $\depth_B(\hat{G}_{\mathcal{U}}-S) \leq d$ if and only if there is a $d$-coloring $\chi$ such that $\left| \{v ~: ~ \max_{b \in \parents(v)} \chi(b) \geq \min_{b \in \children(v)} \chi(b)\} \right| \leq e$, i.e., given a $d$-coloring $\chi$ of the bit vertices, we can simply select $S=\{v ~: ~ \max_{b \in \parents(v)} \chi(b) \geq \min_{b \in \children(v)} \chi(b)\}$ of inconsistent test-vertices and then for every $u \in B$ we can inductively show that $\depth_B(u,\hat{G}_{\mathcal{U}}-S) \leq \chi(u)$.  


\paragraph*{Brief Overview of Svensson's Proof.} 
Svensson defines $\chi(w,i)$ to denote the largest color that is smaller than the colors of at least $(1-\delta)$ fraction of the bit-vertices in $B_{i,w}$, i.e.,  $\chi(w,i)=\max\{\text{color }c\,:\allowbreak\,\underset{b\in B_{i,w}}{\Pr}\left[\chi(b)\ge c\right]\ge1-\delta\}.$
Suppose that it is not possible to satisfy a $\beta=\frac{\delta \eta^2}{t^2 k^2}$-fraction of the constraints in $\mathcal{U}$ for tunable parameters $t,\eta > 0$ that are part of Svensson's construction. 
The core piece of Svensson's proof is demonstrating that if the set $S=\{v ~: ~ \max_{b \in \parents(v)} \chi(b) \geq \min_{b \in \children(v)} \chi(b)\}$ of inconsistent test-vertices has size $|S| \leq (1-32 \delta) |T|$, then we can find some $w \in W$ such that $\Pr[\chi(w,i) > \chi(w,i+1)] \geq 32\delta^2$ for some constant $c$ that depends on various parameters of the construction. 
Svensson notes that by symmetry of the construction $\hat{G}_{\mathcal{U}}$, we can assume without loss of generality that $\chi(i,w) \leq \chi(i+1,w)$ for any $i \leq L$. 
We remark that this will {\em not} necessarily be the case after our indegree reduction step. 
Thus, it immediately follows that $\chi$ uses more than $32|T| \delta^2$ colors, i.e., $\depth_B(u,\hat{G}_{\mathcal{U}}-S) \geq 32|T| \delta^2$.

\subsection{Reducing the Indegree}
As previously discussed, Svensson's construction has indegree that is too large for the purposes of bounding the pebbling complexity by finding a gap between known results implied by $(e_1,d_1)$-reducibility and $(e_2,d_2)$-depth robustness. 
To perform indegree reduction, we use a $\gamma$-extreme depth-robust graph $G_{\gamma,L+1}$ with $L+1$ vertices in a procedure $\bitsparsify_{G_{\gamma,L+1}}(\hat{G}_{\mathcal{U}})$ to decide which edges in $\hat{G}_{\mathcal{U}}$ to keep and which edges to discard. 
Intuitively, we will keep the edge $(b_\ell, t_{\ell'})$ from a bit vertex $b_\ell \in B_\ell$ on layer $\ell \leq \ell'$ to test vertex $t_{\ell'} \in T_{\ell'}$ on layer $\ell'$ if and only if $\ell = \ell'$ or $G_{\gamma,L+1}$ contains the edge $(\ell,\ell')$. 
Similarly, we will keep the edge $(t_\ell, b_{\ell'})$ from a test vertex $t_\ell \in T_{\ell}$ on layer $\ell < \ell'$ to bit vertex $b_{\ell'} \in B_{\ell'}$ on layer $\ell'$ if and only if $(\ell,\ell')\in G_{\gamma,L+1}$. 
The result is a new DAG $\bitsparsify_{G_{\gamma,L+1}}(\hat{G}_{\mathcal{U}})$ with substantially smaller indegree and outdegree $\O{N^{\eps} \log^2 N}$ instead of $\O{N}$. 

\begin{mdframed}
\begin{center}
Transformation $\bitsparsify_{G_{\gamma,L+1}}(\hat{G}_{\mathcal{U}})$
\end{center}
\noindent\underline{Input}: An instance $\hat{G}_{\mathcal{U}}=(V,E)$ of the Svensson's construction, whose vertices are partitioned into $L+1$ bit-layers $B_0,\ldots,B_{L}$ and $L$ test-layers $T_0,\ldots,T_{L-1}$, a $\gamma$-extreme depth robust graph $G_{\gamma,L+1} = (V_\gamma = [L+1],E_\gamma)$. 

\begin{enumerate}[1.]
\item
Let $G'=(V,E)$ be a copy of $\hat{G}_{\mathcal{U}}$. 
\item
If $e=(b,t)$ is an edge in $G$, where $b\in B_{i}$ and $t\in T_{j}$, delete $e$ from $G'$ if $i \neq j$ and $(i,j) \not \in E_\gamma$.
\item
If $e=(t,b)$ is an edge in $G$, where $b\in B_{i}$ and $t\in T_{j}$, delete $e$ from $G'$ if $(j,i) \not \in E_\gamma$.
\end{enumerate}

\noindent
\underline{Output}: $G'$
\end{mdframed}

We remark that we only delete edges from $\hat{G}_{\mathcal{U}}$. 
Thus for any subset $S \subseteq T$ of $|S| \leq e_1$ test vertices, we have $\depth_B(\hat{G}_{\mathcal{U}}-S) \geq \depth_B(\bitsparsify_{G_{\gamma,L+1}}(\hat{G}_{\mathcal{U}}) - S)$. 
Hence, $\bitsparsify_{G_{\gamma,L+1}}(\hat{G}_{\mathcal{U}})$ is certainly not more depth-robust than $\hat{G}_{\mathcal{U}}$. 
The harder argument is showing that the graph $\bitsparsify_{G_{\gamma,L+1}}(\hat{G}_{\mathcal{U}})$ is still depth-robust when our unique games instance $\mathcal{U}$ has no assignment satisfying a $\beta$ fraction of the edges.

Assuming that the Unique Games instance is unsatisfiable, \lemref{lem:svensson:exist} implies that as long as $32\delta^2 |T|$ test-vertices are consistent with our coloring, we can find some $w \in W$ such that $w$ is {\em locally consistent} on at least $32\delta^2 L$ layers, i.e., $w$ is locally consistent on layer $\ell$ if $\forall \ell' > \ell$ we have $\chi(w,\ell')>\chi(w,\ell)$. 

The parameters $\eta, t$ in \lemref{lem:svensson:exist} are tunable parameters of the reduction. 

\newcommand{\lemsvenssonexist}{
Let $\chi$ be any coloring of $\bitsparsify_{G_{\gamma,L+1}}(\hat{G}_{\mathcal{U}})$. If the Unique Games instance has no labeling that satisfies a fraction $\frac{\delta\eta^2}{t^2k^2}$ of the constraints and at least $32\delta^2 |T|$ test vertices are consistent with $\chi$, then there exists $w\in W$ with \[\underset{\ell\in[L]}{\Pr}\left[\chi(w,\ell')>\chi(w,\ell)\text{ for all }\ell'>\ell\text{ with }(\ell,\ell')\in E_{\gamma}\right]\ge 32\delta^2.\]
}
\begin{lemma}
\lemlab{lem:svensson:exist}
\lemsvenssonexist
\end{lemma}
We remark that the proof of \lemref{lem:svensson:exist} closely follows Svensson's argument with a few modifications. 
While the modifications are relatively minor, specifying these modifications requires a complete description of Svensson's construction. 
We refer an interested reader to \appref{app:modified} for details and for the formal proof of \lemref{lem:svensson:exist}. 

\lemref{lem:modified:robust} now shows that $\bitsparsify_{G_{\gamma,L+1}}(\hat{G}_{\mathcal{U}})$ is still depth-robust in case 2. 
The main challenge is that after we sparsify the graph, we can no longer assume that $\chi(w,\ell')>\chi(w,\ell)\text{ for all }\ell'>\ell$ without loss of generality, e.g., even if there are many $i$'s for which $\chi(w,i+1)>\chi(w,i)$ we could have a sequence like $\chi(w,1)=1, \chi(w,2)=2,\chi(w,3)=2, \chi(w,4)=2, \chi(w,5)=1,\chi(w,6)=2,\ldots$. 
We rely on the fact that $G_{\gamma,L+1}$ is extremely depth-robust to show that for any sufficiently large subset $\mathtt{LC} \subseteq [L]$ of layers for which $w$ is {\em locally consistent}, there must be a subsequence  $\mathtt{P}_w \subseteq \mathtt{LC}$ of length $|\mathtt{P}_w | \geq |\mathtt{LC}|-\gamma L$ over which $\chi(w,\cdot)$ is strictly increasing. 

\begin{lemma}
\lemlab{lem:modified:robust}
If the Unique Games instance has no labeling that satisfies a fraction $\frac{\delta\eta^2}{t^2k^2}$ of the constraints and $\gamma\le 31\delta^2$, then for every set $S \subseteq T$ of at most $|S| \leq (1-32\delta)|T|$ test-vertices the graph $G'=\bitsparsify_{G_{\gamma,L+1}}(\hat{G}_{\mathcal{U}})$ has a path of length $\delta^2L$. 
\end{lemma}
\begin{proof}
Suppose the Unique Games instance has no labeling that satisfies a fraction $\frac{\delta\eta^2}{t^2k^2}$ of the constraints. 
Let $S$ contain at most $(1-32\delta)|T|$ and define the labeling $\chi(b) = \depth_B\left(b,G' - S \right)$.  
By \lemref{lem:svensson:exist}, there exists $w\in W$ with 
\[\underset{\ell\in[L]}{\Pr}\left[\chi(w,\ell')>\chi(w,\ell)\text{ for all }\ell'>\ell\text{ with }(\ell,\ell')\in E_{\gamma}\right]\ge32\delta^2.\]
Let $\mathtt{LC} \subseteq [L]$ denote the subset of layers over which $w$ is locally consistent. 
We remark that each $\ell \in \mathtt{LC}$ corresponds to a node in $G_{\gamma,L+1}$ and that $G_{\gamma,L+1}[\mathtt{LC}]$ contains a path $\mathtt{P}_w = (\ell_1,\ldots, \ell_k)$ of length $k \geq \left|\mathtt{LC} \right| - \gamma L \geq (32\delta^2 - \gamma)L$. 
We also note that $\chi(w,\ell_{i+1}) > \chi(w, \ell_{i})$ for each $i < k$. 
Hence, $\chi(w,\ell_k) \geq k$, which means that $\depth_B\left(b,G' - S \right) \geq (32\delta^2 - \gamma)L\geq \delta^2L$ as long as $\gamma\leq 31\delta^2$.
\end{proof}

\thmref{thm:modified:sven}, our main technical result in this section, states that it is Unique Games hard to distinguish between $(e_1,d_1)$-reducible and $(e_2,d_2)$-depth robust graphs even for a DAG $G$ with $N$ vertices and $\indeg(G) = \O{N^{\eps} \log^2 N}$.

\begin{theorem}
\thmlab{thm:modified:sven}
For any integer $k\ge 2$ and constant $\eps>0$, given a DAG $G$ with $N$ vertices and $\indeg(G) = \O{N^{\eps} \log^2 N}$, it is Unique Games hard to distinguish between the following cases: (1) (Completeness): $G$ is $\left(\left(\frac{1-\eps}{k}\right)N,k\right)$-reducible, and (2) (Soundness): $G$ is $\left((1-\eps)N,N^{1-\eps}\right)$-depth robust.
\end{theorem}
\vspace{-5pt}
\begin{proof}
Recall that we can transform $G'= \bitsparsify_{G_{\gamma,L+1}}(\hat{G}_{\mathcal{U}})$ into an unweighted graph $G$ over the $N=|T|$ test-vertices. 
In particular, we add the edge $(u,v)$ to $G$ if and only if there was a path of length $2$ from $u$ to $v$ in $G'$. 
We remark that the indegree is $\indeg(G)=\O{N^{\eps}\log^2 N}$ and that for any $S \subseteq T$, we have $\depth(G-S) \leq  \depth_B( G'-S) \leq \depth(G-S)+1$. 
Completeness now follows immediately from~\thmref{thm:svensson} under the observation that we only removed edges from Svensson's construction. 
Soundness follows immediately from \thmref{thm:svensson} and \lemref{lem:modified:robust}.
\end{proof}

\paragraph*{Obtaining DAGs with Constant Degree.} 
We can now apply a second indegree reduction procedure $\idr(G,\gamma)$. 
For a graph $G=(V,E)$, the procedure $\idr(G,\gamma)$ replaces each node $v\in V$ with a path $P_v = v_1,\ldots,v_{\delta+\gamma}$, where $\delta$ is the indegree of $G$. 
For each edge $(u,v) \in E$, we add the edge $(u_{\delta+\gamma},v_j)$ whenever $(u,v)$ is the $j\th$ incoming edge of $v$, according to some fixed ordering. 
\cite{EC:AlwBloPie17} give parameters $e_2$ and $d_2$ so that $\idr(G,\gamma)$ is $(e_2,d_2)$-depth robust if $G$ is $(e,d)$-depth robust. 
For a formal description of $\idr(G,\gamma)$, see \appref{app:modified}.
We complete the reduction by giving parameters $e_1$ and $d_1$ so that $\idr(G,\gamma)$ is $(e_1,d_1)$-reducible if $G$ is $(e,d)$-reducible. 

\newcommand{\thmindegreductionabp}{There exists a polynomial time procedure $\idr(G,\gamma)$ that takes as input a DAG $G$ with $N$ vertices and $\indeg(G)=\delta$ and outputs a graph $G'=\idr(G,\gamma)$ with $(\delta + \gamma) N$ vertices and $\indeg(G')=2$. 
Moreover, the following properties hold: (1) If $G$ is $(e,d)$-reducible, then $\idr(G,\gamma)$ is $(e,(\delta+\gamma)\cdot d)$-reducible, and (2) If $G$ is $(e,d)$-depth robust, then $\idr(G,\gamma)$ is $(e,\gamma\cdot d)$-depth robust.}
\begin{lemma}
\lemlab{lem:indeg:abp}
\thmindegreductionabp
\end{lemma}

\newcommand{\corsecondidr}{For any integer $k\ge 2$ and constant $\eps>0$, given a DAG $G$ with $N$ vertices and maximum indegree $\indeg(G)=2$, it is Unique Games hard to decide whether $G$ is $(e_1,d_1)$-reducible or $(e_2,d_2)$-depth robust for (Completeness):  $e_1 = \frac{1}{k}N^{\frac{1}{1+2\eps}}$ and $d_1 = kN^{\frac{2\eps}{1+2\eps}}$, and (Soundness): $e_2 = (1-\eps)N^{\frac{1}{1+2\eps}}$ and $d_2 = 0.9N^{\frac{1+\eps}{1+2\eps}}$.}
\begin{corollary}
\corlab{cor:second:idr}
\corsecondidr
\end{corollary}

\section{Putting the Pieces Together}
We would now like to apply \thmref{thm:cc} and \thmref{thm:cc:upper}. 
However, the upper bound on $\pcc(G)$ that we obtain from \thmref{thm:cc} will not be better than $e_1N=\frac{1}{k}N^{\frac{2+2\eps}{1+2\eps}}$, while the lower bound we obtain from \thmref{thm:cc:upper} is just $(1-\eps)N^{\frac{2+\eps}{1+2\eps}}$, so we do not get our desirable gap between the upper and lower bounds. 
We therefore discard \thmref{thm:cc} and \thmref{thm:cc:upper} altogether and instead apply a graph transformation with explicit bounds on pebbling complexity. 

\begin{definition}[Superconcentrator]
A graph $G$ with $\O{N}$ vertices is called a \emph{superconcentrator} if there exists $N$ input vertices, denoted $\inp(G)$, and $N$ output vertices, denoted $\outp(G)$, such that for all $S_1\subseteq \inp(G),S_2\subseteq \outp(G)$ with $|S_1|=|S_2|=k$, there are $k$ vertex disjoint paths from $S_1$ to $S_2$.
\end{definition}
Pippenger gives a superconcentrator construction with depth $\O{\log N}$. 
\begin{lemma}[\cite{Pippenger77}]
\lemlab{lem:pippenger}
There exists a superconcentrator $G$ with at most $42N$ vertices, containing $N$ input vertices and $N$ output vertices, such that $\indeg(G)\le 16$ and $\depth(G)\le\log(42N)$. 
\end{lemma}
Now we define the overlay of a superconcentrator on a graph $G$ (see \figref{scoverlay}).

\begin{definition}[Superconcentrator Overlay]\deflab{def:superc_overlay}
Let $G=(V(G),E(G))$ be a fixed DAG with $N$ vertices and $G_S=(V(G_S),E(G_S))$ be a (priori fixed) superconcentrator with $N$ input vertices $\inp(G_S)=\{i_1,\cdots,i_N\}\subseteq V(G_S)$ and $N$ output vertices $\outp(G_S)=\{o_1,\cdots,o_N\}\subseteq V(G_S)$. We call a graph $G'=(V(G_S),E(G_S)\cup E_I \cup E_O)$ a \emph{superconcentrator overlay} where $E_I=\{(i_u,i_v):(u,v)\in E(G)\}$ and $E_O=\{(o_i,o_{i+1}):1\leq i<N\}$ and denote as $G'=\superconc(G)$.
\end{definition}
We will denote the interior nodes as $\inte(G')=G'\setminus(\inp(G')\cup\outp(G'))$ where $\inp(G')=\inp(G_S)$ and $\outp(G')=\outp(G_S)$. 
We remark that when using Pippenger's construction of superconcentrators, it is easy to show that $\superconc(G)$ is\\$\left(e+\frac{N}{d},2d+\log(42N)\right)$-reducible whenever $G$ is $(e,d)$-reducible, which implies that
\[\pcc(\superconc(G))\le\underset{g\ge d}{\min}\,\left(e+\frac{N}{d}\right)42N+2g(42N)+\frac{42N}{g}\left(2d+\log(42N)\right)42N.\]
For more details, we refer an interested reader to \lemref{lem:overlay:ed} and \corref{cor:sc(G)naive} in \appref{missing}. 
However, these results are not quite as strong as we would like.  
By comparison, we have the following lower bound on the pebbling complexity from \cite{EPRINT:BHKLXZ18}: 
\[\pcc(\superconc(G))\ge\min\left(\frac{eN}{8},\frac{dN}{8}\right).\]

In \lemref{lem:sc(G)improved} we obtain a {\em significantly tighter} upper bound on $\pcc(\superconc(G))$ with an improved pebbling strategy described at the end of this section. 
 
 \begin{lemma}\lemlab{lem:sc(G)improved}
Let $G$ be an $(e,d)$-reducible graph with $N$ vertices with $\indeg(G)=2$. 
Then $\pcc(\superconc(G))\le\underset{g\ge d}{\min}\,\left\{2eN+4gN+\frac{43dN^2}{g}+\frac{24N^2\log(42N)}{g}+42N\log(42N)+N\right\}.$
\end{lemma}
 
With the improved attack in \lemref{lem:sc(G)improved}, we can tune parameters appropriately to obtain our main result, \thmref{thm:approx}.
\begin{theorem}
\thmlab{thm:approx}
Given a DAG $G$, it is Unique Games hard to approximate $\pcc(G)$ within any constant factor.
\end{theorem}
\begin{proof}
Let $k\ge 2$ be an integer that we shall later fix and similarly, let $\eps>0$ be a constant that we will later fix. 
Given a DAG $G$ with $N$ vertices, then it follows by \corref{cor:second:idr} that it is Unique Games hard to decide whether $G$ is $(e_1,d_1)$-reducible or $(e_2,d_2)$-depth robust for $e_1 = \frac{1}{k}N^{\frac{1}{1+2\eps}}$, $d_1 = kN^{\frac{2\eps}{1+2\eps}}$ and $e_2 = (1-\eps)N^{\frac{1}{1+2\eps}}$ and $d_2 = 0.9 N^{\frac{1+\eps}{1+2\eps}}$. 
If $G$ is $(e_1,d_1)$-reducible, then by \lemref{lem:sc(G)improved}, 
$\pcc(\superconc(G))\le\underset{g\ge d}{\min}\,\{2e_1N+4gN+\frac{43d_1N^2}{g}+\frac{24N^2\log(42N)}{g}+42N\log(42N)+N\}.$
Observe that $2e_1N=\frac{2}{k}N^{(2+2\eps)/(1+2\eps)}$, whereas for $g=e_1$ and sufficiently large $N$,
$4gN+\frac{43d_1N^2}{g}+\frac{24N^2\log(42N)}{g}+42N\log(42N)+N\le\frac{5}{k}N^{\frac{2+2\eps}{1+2\eps}}.$
Hence for sufficiently large $N$,
\[\pcc(\superconc(G))\le\frac{7}{k}N^{\frac{2+2\eps}{1+2\eps}}.\]
On the other hand, if $G$ is $(e_2,d_2)$-depth robust, then by \lemref{lem:superconc:cc:lower},
$\pcc(\superconc(G))\ge\min\left(\frac{e_2N}{8},\frac{d_2N}{8}\right).$
Specifically, 
\[\pcc(\superconc(G))\ge\frac{e_2N}{8}=\frac{1-\eps}{8}N^{\frac{2+2\eps}{1+2\eps}}.\]

Let $c > 1$ be any constant. 
Setting $\eps=0.1$ and $k=\lceil\frac{560}{9}c^2\rceil$, we get that if $G$ is $(e_1,d_1)$-reducible, then $\pcc(\superconc(G))\le\frac{9}{80c^2}N^{\frac{2+2\eps}{1+2\eps}}$ but if $G$ is $(e_2,d_2)$-reducible, then $\pcc(\superconc(G))\ge\frac{9}{80}N^{\frac{2+2\eps}{1+2\eps}}$. 
Hence, it is Unique Games hard to approximate $\pcc(G)$ with a factor of $c$. 
\end{proof}

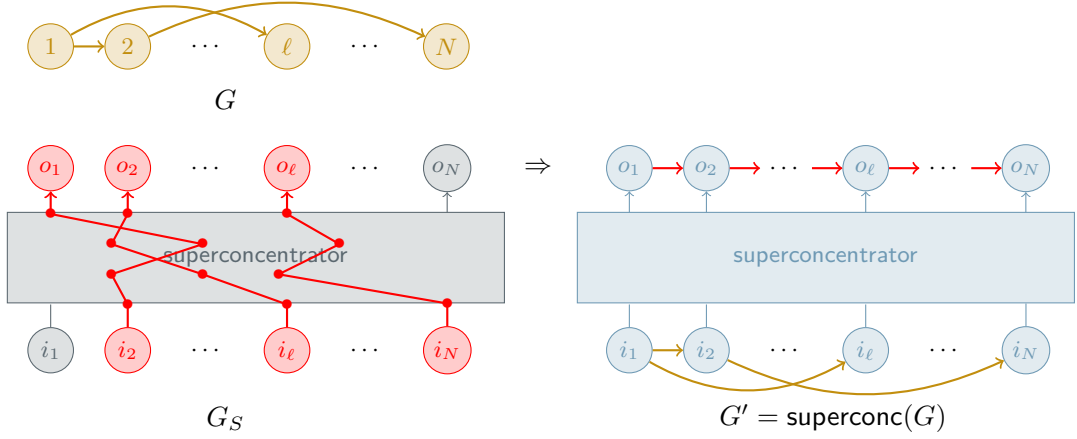
\begin{figure}
\centering
\begin{tikzpicture}
    [node distance=1cm,auto,font=\small,
    every node/.style={node distance=2cm},
    bit/.style={circle,draw,purdueevertrueblue,fill=purdueevertrueblue!20,text=purdueevertrueblue,inner sep=2pt, minimum width=0.6cm},
    test/.style={circle,draw,purduecampusgold,fill=purduecampusgold!20,text=purduecampusgold,inner sep=2pt, minimum width=0.6cm},
    sub/.style={circle,draw,red,fill=red!20,text=red,inner sep=2pt, minimum width=0.6cm},
    overlay/.style={circle,draw,purdueslayterskyblue,fill=purdueslayterskyblue!20,text=purdueslayterskyblue,inner sep=2pt, minimum width=0.6cm},
    force/.style={rectangle,draw,purdueevertrueblue,fill=purdueevertrueblue!20,text=purdueevertrueblue,inner sep=2pt,text width=6.4cm,text badly centered,minimum height=1.2cm,font=\small},
    overforce/.style={rectangle,draw,purdueslayterskyblue,fill=purdueslayterskyblue!20,text=purdueslayterskyblue,inner sep=2pt,text width=6.4cm,text badly centered,minimum height=1.2cm,font=\small}
    ]
    \node[test] (1) at (0,0) {$1$};
    \node[test,right=0.4cm of 1] (2) {$2$};
    \node[right=0.4cm of 2] (cd1) {$\cdots$};
    \node[test,right=0.4cm of cd1] (l) {$\ell$};
    \node[right=0.4cm of l] (cd2) {$\cdots$};
    \node[test,right=0.4cm of cd2] (N) {$N$};
    \node (G) at (2.3,-0.7) {\large $G$};
    \path[draw,thick,color=purduecampusgold,->] (1) edge (2);
    \path[draw,thick,color=purduecampusgold,->] (1) edge[bend left=30] (l);
    \path[draw,thick,color=purduecampusgold,->] (2) edge[bend left=25] (N);
    
    \node[sub,below=1cm of 1] (o1) {$o_1$};
    \node[sub,right=0.4cm of o1] (o2) {$o_2$};
    \node[right=0.4cm of o2] (cd3) {$\cdots$};
    \node[sub,right=0.4cm of cd3] (ol) {$o_{\ell}$};
    \node[right=0.4cm of ol] (cd4) {$\cdots$};
    \node[bit,right=0.4cm of cd4] (oN) {$o_N$};
    \node[coordinate,below=0.3cm of o1] (a1) {};
    \node[coordinate,below=0.3cm of o2] (a2) {};
    \node[coordinate,below=0.3cm of ol] (al) {};
    \node[coordinate,below=0.3cm of oN] (aN) {};
    \path[draw,thick,color=red,->] (a1) -- (o1);
    \path[draw,thick,color=red,->] (a2) -- (o2);
    \path[draw,thick,color=red,->] (al) -- (ol);
    \path[draw,color=purdueevertrueblue,->] (aN) -- (oN);
    \node[force] (sc) at (2.7,-2.8) {\textsf{superconcentrator}};
    \node at (a1) {\textcolor{red}{$\bullet$}};
    \node at (a2) {\textcolor{red}{$\bullet$}};
    \node at (al) {\textcolor{red}{$\bullet$}};
    \node[bit,below=1.8cm of o1] (i1) {$i_1$};
    \node[sub,right=0.4cm of i1] (i2) {$i_2$};
    \node[right=0.4cm of i2] (cd5) {$\cdots$};
    \node[sub,right=0.4cm of cd5] (il) {$i_{\ell}$};
    \node[right=0.4cm of il] (cd6) {$\cdots$};
    \node[sub,right=0.4cm of cd6] (iN) {$i_N$};
    \node[coordinate,above=0.3cm of i1] (b1) {};
    \node[coordinate,above=0.3cm of i2] (b2) {};
    \node[coordinate,above=0.3cm of il] (bl) {};
    \node[coordinate,above=0.3cm of iN] (bN) {};
    \node at (b2) {\textcolor{red}{$\bullet$}};
    \node at (bl) {\textcolor{red}{$\bullet$}};
    \node at (bN) {\textcolor{red}{$\bullet$}};
    \path[draw,color=purdueevertrueblue] (b1) -- (i1);
    \path[draw,thick,color=red] (b2) -- (i2);
    \path[draw,thick,color=red] (bl) -- (il);
    \path[draw,thick,color=red] (bN) -- (iN);
    \node[coordinate,below=0.4cm of a1] (c) {$\bullet$};
    \node[coordinate,below=0.4cm of c] (d) {};
    \node[coordinate,right=2cm of c] (c1) {};
    \node[coordinate,left=1.2cm of c1] (c2) {};
    \node[coordinate,right=3cm of c2] (cl) {};
    \node at (c1) {\textcolor{red}{$\bullet$}};
    \node at (c2) {\textcolor{red}{$\bullet$}};
    \node at (cl) {\textcolor{red}{$\bullet$}};
    \node[coordinate,right=0.8cm of d] (d1) {};
    \node[coordinate,right=1.2cm of d1] (d2) {};
    \node[coordinate,right=1cm of d2] (dl) {};
    \node at (d1) {\textcolor{red}{$\bullet$}};
    \node at (d2) {\textcolor{red}{$\bullet$}};
    \node at (dl) {\textcolor{red}{$\bullet$}};
    \path[draw,thick,color=red] (a1) -- (c1);
    \path[draw,thick,color=red] (a2) -- (c2);
    \path[draw,thick,color=red] (al) -- (cl);
    \path[draw,thick,color=red] (c1) -- (d1);
    \path[draw,thick,color=red] (c2) -- (d2);
    \path[draw,thick,color=red] (cl) -- (dl);
    \path[draw,thick,color=red] (b2) -- (d1);
    \path[draw,thick,color=red] (bl) -- (d2);
    \path[draw,thick,color=red] (bN) -- (dl);
    \node[below=3.7cm of G] (GS) {\large $G_S$};
    
    \node[right=5.8cm of o1] (right) {\large $\Rightarrow$};
    \node[overlay,right=7cm of o1] (sco1) {$o_1$};
    \node[overlay,right=0.4cm of sco1] (sco2) {$o_2$};
    \node[right=0.4cm of sco2] (sccd3) {$\cdots$};
    \node[overlay,right=0.4cm of sccd3] (scol) {$o_{\ell}$};
    \node[right=0.4cm of scol] (sccd4) {$\cdots$};
    \node[overlay,right=0.4cm of sccd4] (scoN) {$o_N$};
    \node[coordinate,below=0.3cm of sco1] (sca1) {};
    \node[coordinate,below=0.3cm of sco2] (sca2) {};
    \node[coordinate,below=0.3cm of scol] (scal) {};
    \node[coordinate,below=0.3cm of scoN] (scaN) {};
    \path[draw,color=purdueslayterskyblue,->] (sca1) -- (sco1);
    \path[draw,color=purdueslayterskyblue,->] (sca2) -- (sco2);
    \path[draw,color=purdueslayterskyblue,->] (scal) -- (scol);
    \path[draw,color=purdueslayterskyblue,->] (scaN) -- (scoN);
    \node[overforce] (scsc) at (10.2,-2.8) {\textsf{superconcentrator}};
    \node[overlay,below=1.8cm of sco1] (sci1) {$i_1$};
    \node[overlay,right=0.4cm of sci1] (sci2) {$i_2$};
    \node[right=0.4cm of sci2] (sccd5) {$\cdots$};
    \node[overlay,right=0.4cm of sccd5] (scil) {$i_{\ell}$};
    \node[right=0.4cm of scil] (sccd6) {$\cdots$};
    \node[overlay,right=0.4cm of sccd6] (sciN) {$i_N$};
    \node[coordinate,above=0.3cm of sci1] (scb1) {};
    \node[coordinate,above=0.3cm of sci2] (scb2) {};
    \node[coordinate,above=0.3cm of scil] (scbl) {};
    \node[coordinate,above=0.3cm of sciN] (scbN) {};
    \path[draw,color=purdueslayterskyblue] (scb1) -- (sci1);
    \path[draw,color=purdueslayterskyblue] (scb2) -- (sci2);
    \path[draw,color=purdueslayterskyblue] (scbl) -- (scil);
    \path[draw,color=purdueslayterskyblue] (scbN) -- (sciN);
    \path[draw,thick,color=red,->] (sco1) -- (sco2);
    \path[draw,thick,color=red,->] (sco2) -- (sccd3);
    \path[draw,thick,color=red,->] (sccd3) -- (scol);
    \path[draw,thick,color=red,->] (scol) -- (sccd4);
    \path[draw,thick,color=red,->] (sccd4) -- (scoN);
    \path[draw,thick,color=purduecampusgold,->] (sci1) edge (sci2);
    \path[draw,thick,color=purduecampusgold,->] (sci1) edge[bend right=30] (scil);
    \path[draw,thick,color=purduecampusgold,->] (sci2) edge[bend right=25] (sciN);
    \node[right=6cm of GS] (superconc) {\large $G'=\mathsf{superconc}(G)$};
\end{tikzpicture}
\caption{An example of the superconcentrator overlay $G'=\mathsf{superconc}(G)$. We remark that as illustrated in $G_S$, for any $k$ inputs and $k$ outputs (highlighted in red), there exist $k$ vertex disjoint paths in a superconcentrator.}
\figlab{scoverlay}
\end{figure}

\begin{figure}
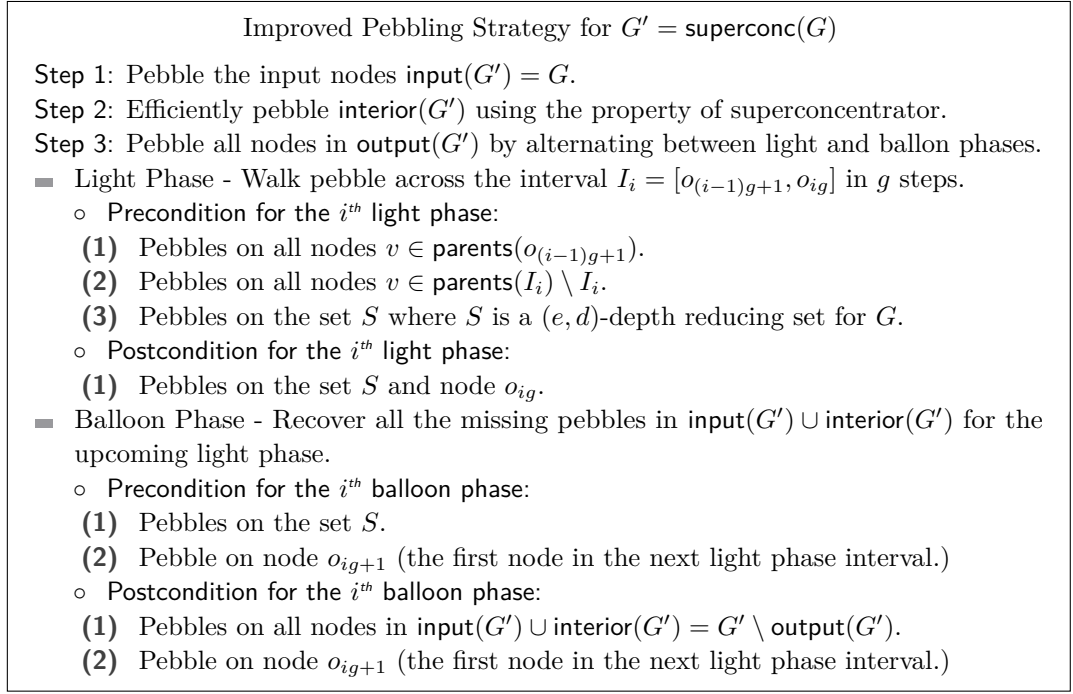

\begin{mdframed}
\begin{center}
Improved Pebbling Strategy for $G'=\superconc(G)$
\end{center}
\vskip 0.5em
\textsf{Step 1:} Pebble the input nodes $\inp(G')=G$.

\textsf{Step 2:} Efficiently pebble $\inte(G')$ using the property of superconcentrator.

\textsf{Step 3:} Pebble all nodes in $\outp(G')$ by alternating between light and ballon phases.
\begin{itemize}
\item Light Phase - Walk pebble across the interval $I_i=[o_{(i-1)g+1},o_{ig}]$ in $g$ steps.
\begin{itemize}
\item \textsf{Precondition for the $i\th$ light phase:}
\begin{enumerate}
\item Pebbles on all nodes $v\in\parents(o_{(i-1)g+1})$.
\item Pebbles on all nodes $v\in\parents(I_i)\setminus I_i$.
\item Pebbles on the set $S$ where $S$ is a $(e,d)$-depth reducing set for $G$.
\end{enumerate}
\item \textsf{Postcondition for the $i\th$ light phase:}
\begin{enumerate}
\item Pebbles on the set $S$ and node $o_{ig}$.
\end{enumerate}
\end{itemize}
\item Balloon Phase - Recover all the missing pebbles in $\inp(G')\cup\inte(G')$ for the upcoming light phase.
\begin{itemize}
\item \textsf{Precondition for the $i\th$ balloon phase:}
\begin{enumerate}
\item Pebbles on the set $S$.
\item Pebble on node $o_{ig+1}$ (the first node in the next light phase interval.)
\end{enumerate}
\item \textsf{Postcondition for the $i\th$ balloon phase:}
\begin{enumerate}
\item Pebbles on all nodes in $\inp(G')\cup\inte(G')=G'\setminus\outp(G')$.
\item Pebble on node $o_{ig+1}$ (the first node in the next light phase interval.)
\end{enumerate}
\end{itemize}
\end{itemize}
\end{mdframed}
\caption{An improved pebbling strategy for $G'=\superconc(G)$. It brought ideas from~\cite{C:AlwBlo16}.}
\figlab{pebbling}
\end{figure}

\newpage

\begin{proof}[Proof of \lemref{lem:sc(G)improved}]
We will examine the pebbling cost of $\superconc(G)$ for each step shown in \figref{pebbling}.

\begin{itemize}
\item \textsf{Step 1:} 
We need to place pebbles on all input nodes in $G'$. 
By \thmref{thm:cc}, the pebbling cost of $\inp(G')=G$ will be upper bounded by\footnote{Note that \thmref{thm:cc} shows the upper bound of the pebbling cost to pebble the last node of $G$. Here, the difference is that we have to pebble all nodes in $\inp(G')=G$, not the last node of $G$ only. However, \cite{C:AlwBlo16} says that we can recover all nodes concurrently by running one more balloon phase and such cost is already contained in the term $N^2d/g$. Therefore, we have the same upper bound for $\pcc(G)$.}
\[ \pcc(G)\le\underset{g\ge d}{\min}\, \left\{eN+2gN+\frac{N^2d}{g}\right\}. \]
\item \textsf{Step 2:} 
Start with a configuration with pebbles on every node in $\inp(G')$. 
We have that $\depth(G')\setminus\inp(G')=\log(42N)$. 
Therefore, in time $\log(42N)$, we can place pebbles on every node in $\inp(G')\cup\inte(G')$. 
Hence, the total pebbling cost in Step 2 will be at most $42N\log(42N)$.
\item \textsf{Step 3:} 
The goal for step 3 is to walk a pebble across the output nodes starting from $o_1$ to $o_N$. 
To save cost during this step, we should alternate light phases and balloon phases repeatedly $N/g$ times in total since we walk pebble across the interval $I_i=[o_{(i-1)g+1},o_{ig}]$ of length $g$ in $\outp(G')$ in each phase. 
Let $S$ be a $(e,d)$-depth reducing set for $G$. 
In each light phase, to walk a pebble across the interval $I_i$, we should keep pebbles on $S$ and $\parents(I_i)\setminus I_i$. 
Since each node in $I_i$ has two parents outside the interval and we keep one pebble in $I_i$ (the current node) for each step, the maximum number of pebbles to keep would be $|S|+2g+1=e+2g+1$ for each step. 
Hence, the maximum pebbling cost to walk pebble across $I_i$ in $i\th$ light phase is $(e+2g+1)g$. 
In each balloon phase, we recover the pebbles in $\inp(G')\cup\inte(G')$ for the next light phase. 
Since $S$ is a $(e,d)$-depth reducing set, we have that $\depth(G'\setminus(S\cup\outp(G')))\leq d+\log(42N)$. 
Therefore, recovering the pebbles will cost at most $(d+\log(42N))42N$ for each balloon phase. 
Hence, the total pebbling cost for Step 3 will be at most $\left[ (e+2g+1)g + (d+\log(42N))42N \right]\frac{N}{g}$.
\end{itemize}

\noindent Taken together, we have that
\begin{scriptsize}
\begin{align*}
\pcc(\superconc_2(G)) &\le \underset{g\ge d}{\min}\, \bigg\{\underbrace{eN+2gN+\frac{N^2d}{g}}_{\mathsf{Step~1}} + \underbrace{42N\log(42N)}_{\mathsf{Step~2}} +\underbrace{\left[ (e+2g+1)g+(d+\log(42N)42N\right]\frac{N}{g}}_{\mathsf{Step~3}} \bigg\}\\
&\le \underset{g\ge d}{\min}\,\left\{2eN+4gN+\frac{43dN^2}{g}+\frac{24N^2\log(42N)}{g}+42N\log(42N)+N\right\}
\end{align*}
\end{scriptsize}
as desired.
\end{proof}

\def\shortbib{0}
\bibliography{references,abbrev3,smallcite}

\begin{thebibliography}{10}

\bibitem{NDSS:AbadiBW03}
Mart{\'{\i}}n Abadi, Michael Burrows, and Ted Wobber.
\newblock Moderately hard, memory-bound functions.
\newblock In {\em Proceedings of the Network and Distributed System Security
  Symposium, {NDSS} 2003, San Diego, California, {USA}}, 2003.

\bibitem{C:AlwBlo16}
Jo{\"e}l Alwen and Jeremiah Blocki.
\newblock Efficiently computing data-independent memory-hard functions.
\newblock In Matthew Robshaw and Jonathan Katz, editors, {\em CRYPTO~2016,
  Part~II}, volume 9815 of {\em {LNCS}}, pages 241--271. Springer, Heidelberg,
  August 2016.
\newblock \href {http://dx.doi.org/10.1007/978-3-662-53008-5_9}
  {\path{doi:10.1007/978-3-662-53008-5_9}}.

\bibitem{ESP:AlwBlo17}
Jo{\"e}l Alwen and Jeremiah Blocki.
\newblock Towards practical attacks on argon2i and balloon hashing.
\newblock In {\em Security and Privacy (EuroS\&P), 2017 IEEE European Symposium
  on}, pages 142--157. IEEE, 2017.

\bibitem{CCS:AlwBloHar17}
Jo{\"e}l Alwen, Jeremiah Blocki, and Ben Harsha.
\newblock Practical graphs for optimal side-channel resistant memory-hard
  functions.
\newblock In Bhavani~M. Thuraisingham, David Evans, Tal Malkin, and Dongyan Xu,
  editors, {\em ACM CCS 2017}, pages 1001--1017. {ACM} Press,
  October~/~November 2017.
\newblock \href {http://dx.doi.org/10.1145/3133956.3134031}
  {\path{doi:10.1145/3133956.3134031}}.

\bibitem{EC:AlwBloPie17}
Jo{\"e}l Alwen, Jeremiah Blocki, and Krzysztof Pietrzak.
\newblock Depth-robust graphs and their cumulative memory complexity.
\newblock In Jean{-}S{\'{e}}bastien Coron and Jesper~Buus Nielsen, editors,
  {\em EUROCRYPT~2017, Part~III}, volume 10212 of {\em {LNCS}}, pages 3--32.
  Springer, Heidelberg, April~/~May 2017.
\newblock \href {http://dx.doi.org/10.1007/978-3-319-56617-7_1}
  {\path{doi:10.1007/978-3-319-56617-7_1}}.

\bibitem{EC:AlwBloPie18}
Jo{\"e}l Alwen, Jeremiah Blocki, and Krzysztof Pietrzak.
\newblock Sustained space complexity.
\newblock In Jesper~Buus Nielsen and Vincent Rijmen, editors, {\em
  EUROCRYPT~2018, Part~II}, volume 10821 of {\em {LNCS}}, pages 99--130.
  Springer, Heidelberg, April~/~May 2018.
\newblock \href {http://dx.doi.org/10.1007/978-3-319-78375-8_4}
  {\path{doi:10.1007/978-3-319-78375-8_4}}.

\bibitem{STOC:AlwSer15}
Jo{\"e}l Alwen and Vladimir Serbinenko.
\newblock High parallel complexity graphs and memory-hard functions.
\newblock In Rocco~A. Servedio and Ronitt Rubinfeld, editors, {\em 47th ACM
  STOC}, pages 595--603. {ACM} Press, June 2015.
\newblock \href {http://dx.doi.org/10.1145/2746539.2746622}
  {\path{doi:10.1145/2746539.2746622}}.

\bibitem{FOCS:AroBarSte10}
Sanjeev Arora, Boaz Barak, and David Steurer.
\newblock Subexponential algorithms for unique games and related problems.
\newblock In {\em 51st FOCS}, pages 563--572. {IEEE} Computer Society Press,
  October 2010.
\newblock \href {http://dx.doi.org/10.1109/FOCS.2010.59}
  {\path{doi:10.1109/FOCS.2010.59}}.

\bibitem{BansalK09}
Nikhil Bansal and Subhash Khot.
\newblock Optimal long code test with one free bit.
\newblock In {\em 50th Annual {IEEE} Symposium on Foundations of Computer
  Science, {FOCS}}, pages 453--462, 2009.

\bibitem{blocki2019relaxed}
Jeremiah Blocki, Venkata Gandikota, Elena Grigorescu, and Samson Zhou.
\newblock Relaxed locally correctable code in computationally bounded channels.
\newblock In {\em IEEE International Symposium on Information Theory (ISIT)},
  2019.

\bibitem{EPRINT:BHKLXZ18}
Jeremiah Blocki, Benjamin Harsha, Siteng Kang, Seunghoon Lee, Lu~Xing, and
  Samson Zhou.
\newblock Data-independent memory hard functions: New attacks and stronger
  constructions.
\newblock {\em {IACR} Cryptology ePrint Archive}, 2018:944, 2018.

\bibitem{BlockiKZ19}
Jeremiah Blocki, Shubhang Kulkarni, and Samson Zhou.
\newblock On locally decodable codes in resource bounded channels.
\newblock {\em CoRR}, abs/1909.11245, 2019.

\bibitem{CCS:BloRenZho18}
Jeremiah Blocki, Ling Ren, and Samson Zhou.
\newblock Bandwidth-hard functions: Reductions and lower bounds.
\newblock In David Lie, Mohammad Mannan, Michael Backes, and XiaoFeng Wang,
  editors, {\em ACM CCS 2018}, pages 1820--1836. {ACM} Press, October 2018.
\newblock \href {http://dx.doi.org/10.1145/3243734.3243773}
  {\path{doi:10.1145/3243734.3243773}}.

\bibitem{TCC:BloZho17}
Jeremiah Blocki and Samson Zhou.
\newblock On the depth-robustness and cumulative pebbling cost of {Argon2i}.
\newblock In Yael Kalai and Leonid Reyzin, editors, {\em TCC~2017, Part~I},
  volume 10677 of {\em {LNCS}}, pages 445--465. Springer, Heidelberg, November
  2017.
\newblock \href {http://dx.doi.org/10.1007/978-3-319-70500-2_15}
  {\path{doi:10.1007/978-3-319-70500-2_15}}.

\bibitem{FC:BloZho18}
Jeremiah Blocki and Samson Zhou.
\newblock On the computational complexity of minimal cumulative cost graph
  pebbling.
\newblock {\em Financial Cryptography and Data Security (FC 2018)}, 2018.

\bibitem{EPRINT:CecMieJue18}
Ethan Cecchetti, Ian Miers, and Ari Juels.
\newblock {PIEs}: Public incompressible encodings for decentralized storage.
\newblock Cryptology ePrint Archive, Report 2018/684, 2018.
\newblock \url{https://eprint.iacr.org/2018/684}.

\bibitem{Cha73}
Ashok~K. Chandra.
\newblock Efficient compilation of linear recursive programs.
\newblock In {\em SWAT (FOCS)}, pages 16--25, 1973.

\bibitem{CharikarGM09}
Moses Charikar, Venkatesan Guruswami, and Rajsekar Manokaran.
\newblock Every permutation {CSP} of arity 3 is approximation resistant.
\newblock In {\em Proceedings of the 24th Annual {IEEE} Conference on
  Computational Complexity, {CCC} 2009, Paris, France, 15-18 July 2009}, pages
  62--73, 2009.

\bibitem{Coo73}
Stephen~A. Cook.
\newblock An observation on time-storage trade off.
\newblock In {\em Proceedings of the Fifth Annual ACM Symposium on Theory of
  Computing}, STOC '73, pages 29--33, 1973.

\bibitem{DemaineL17}
Erik~D. Demaine and Quanquan~C. Liu.
\newblock Inapproximability of the standard pebble game and hard to pebble
  graphs.
\newblock In {\em Algorithms and Data Structures - 15th International
  Symposium, {WADS} 2017, St. John's, NL, Canada, July 31 - August 2, 2017,
  Proceedings}, pages 313--324, 2017.

\bibitem{DinurKKMS18}
Irit Dinur, Subhash Khot, Guy Kindler, Dor Minzer, and Muli Safra.
\newblock Towards a proof of the 2-to-1 games conjecture?
\newblock In {\em Proceedings of the 50th Annual {ACM} {SIGACT} Symposium on
  Theory of Computing, {STOC} 2018, Los Angeles, CA, USA, June 25-29, 2018},
  pages 376--389, 2018.

\bibitem{DGN03}
Cynthia Dwork, Andrew Goldberg, and Moni Naor.
\newblock On memory-bound functions for fighting spam.
\newblock In {\em Advances in Cryptology - CRYPTO 2003, 23rd Annual
  International Cryptology Conference, Santa Barbara, California, USA, August
  17-21, 2003, Proceedings}, volume 2729 of {\em Lecture Notes in Computer
  Science}, pages 426--444. Springer, 2003.
\newblock URL:
  \url{http://www.iacr.org/cryptodb/archive/2003/CRYPTO/1266/1266.pdf}.

\bibitem{C:DwoNaoWee05}
Cynthia Dwork, Moni Naor, and Hoeteck Wee.
\newblock Pebbling and proofs of work.
\newblock In Victor Shoup, editor, {\em CRYPTO~2005}, volume 3621 of {\em
  {LNCS}}, pages 37--54. Springer, Heidelberg, August 2005.
\newblock \href {http://dx.doi.org/10.1007/11535218_3}
  {\path{doi:10.1007/11535218_3}}.

\bibitem{C:DFKP15}
Stefan Dziembowski, Sebastian Faust, Vladimir Kolmogorov, and Krzysztof
  Pietrzak.
\newblock Proofs of space.
\newblock In Rosario Gennaro and Matthew J.~B. Robshaw, editors, {\em
  CRYPTO~2015, Part~II}, volume 9216 of {\em {LNCS}}, pages 585--605. Springer,
  Heidelberg, August 2015.
\newblock \href {http://dx.doi.org/10.1007/978-3-662-48000-7_29}
  {\path{doi:10.1007/978-3-662-48000-7_29}}.

\bibitem{EC:Fisch19}
Ben Fisch.
\newblock Tight proofs of space and replication.
\newblock In Yuval Ishai and Vincent Rijmen, editors, {\em EUROCRYPT~2019,
  Part~II}, volume 11477 of {\em {LNCS}}, pages 324--348. Springer, Heidelberg,
  May 2019.
\newblock \href {http://dx.doi.org/10.1007/978-3-030-17656-3_12}
  {\path{doi:10.1007/978-3-030-17656-3_12}}.

\bibitem{AC:ForLucWen14}
Christian Forler, Stefan Lucks, and Jakob Wenzel.
\newblock Memory-demanding password scrambling.
\newblock In Palash Sarkar and Tetsu Iwata, editors, {\em ASIACRYPT~2014,
  Part~II}, volume 8874 of {\em {LNCS}}, pages 289--305. Springer, Heidelberg,
  December 2014.
\newblock \href {http://dx.doi.org/10.1007/978-3-662-45608-8_16}
  {\path{doi:10.1007/978-3-662-45608-8_16}}.

\bibitem{GilbertLT79}
John~R. Gilbert, Thomas Lengauer, and Robert~Endre Tarjan.
\newblock The pebbling problem is complete in polynomial space.
\newblock In {\em Proceedings of the 11h Annual {ACM} Symposium on Theory of
  Computing (STOC)}, pages 237--248, 1979.

\bibitem{STOC:GoeWil94}
Michel~X. Goemans and David~P. Williamson.
\newblock .879-approximation algorithms for {MAX} {CUT} and {MAX} {2SAT}.
\newblock In {\em 26th ACM STOC}, pages 422--431. {ACM} Press, May 1994.
\newblock \href {http://dx.doi.org/10.1145/195058.195216}
  {\path{doi:10.1145/195058.195216}}.

\bibitem{HP70}
Carl~E. Hewitt and Michael~S. Paterson.
\newblock Record of the project mac conference on concurrent systems and
  parallel computation, 1970.

\bibitem{HongK81}
Jia{-}Wei Hong and H.~T. Kung.
\newblock {I/O} complexity: The red-blue pebble game.
\newblock In {\em Proceedings of the 13th Annual {ACM} Symposium on Theory of
  Computing, May 11-13, 1981, Milwaukee, Wisconsin, {USA}}, pages 326--333,
  1981.

\bibitem{STOC:Khot02b}
Subhash Khot.
\newblock On the power of unique 2-prover 1-round games.
\newblock In {\em 34th ACM STOC}, pages 767--775. {ACM} Press, May 2002.
\newblock \href {http://dx.doi.org/10.1145/509907.510017}
  {\path{doi:10.1145/509907.510017}}.

\bibitem{KhotKMO07}
Subhash Khot, Guy Kindler, Elchanan Mossel, and Ryan O'Donnell.
\newblock Optimal inapproximability results for {MAX-CUT} and other 2-variable
  csps?
\newblock {\em {SIAM} J. Comput.}, 37(1):319--357, 2007.

\bibitem{KhotMS17}
Subhash Khot, Dor Minzer, and Muli Safra.
\newblock On independent sets, 2-to-2 games, and grassmann graphs.
\newblock In {\em Proceedings of the 49th Annual {ACM} {SIGACT} Symposium on
  Theory of Computing, {STOC} 2017, Montreal, QC, Canada, June 19-23, 2017},
  pages 576--589, 2017.

\bibitem{KhotR08}
Subhash Khot and Oded Regev.
\newblock Vertex cover might be hard to approximate to within 2-epsilon.
\newblock {\em J. Comput. Syst. Sci.}, 74(3):335--349, 2008.

\bibitem{LT82}
Thomas Lengauer and Robert~E. Tarjan.
\newblock Asymptotically tight bounds on time-space trade-offs in a pebble
  game.
\newblock {\em J. ACM}, 29(4):1087--1130, October 1982.

\bibitem{Liu17}
Quanquan Liu.
\newblock Red-blue and standard pebble games: Complexity and applications in
  the sequential and parallel models.
\newblock Master's thesis, Massachusetts Institute of Technology, Feburary
  2017.
\newblock URL: \url{http://erikdemaine.org/theses/qliuM.pdf}.

\bibitem{ITCS:MahMorVad13}
Mohammad Mahmoody, Tal Moran, and Salil~P. Vadhan.
\newblock Publicly verifiable proofs of sequential work.
\newblock In Robert~D. Kleinberg, editor, {\em ITCS 2013}, pages 373--388.
  {ACM}, January 2013.
\newblock \href {http://dx.doi.org/10.1145/2422436.2422479}
  {\path{doi:10.1145/2422436.2422479}}.

\bibitem{Nordstrom13}
Jakob Nordstr{\"{o}}m.
\newblock Pebble games, proof complexity, and time-space trade-offs.
\newblock {\em Logical Methods in Computer Science}, 9(3), 2013.

\bibitem{ITCS:Pietrzak19a}
Krzysztof Pietrzak.
\newblock Proofs of catalytic space.
\newblock In Avrim Blum, editor, {\em ITCS 2019}, volume 124, pages
  59:1--59:25. {LIPIcs}, January 2019.
\newblock \href {http://dx.doi.org/10.4230/LIPIcs.ITCS.2019.59}
  {\path{doi:10.4230/LIPIcs.ITCS.2019.59}}.

\bibitem{Pippenger77}
Nicholas Pippenger.
\newblock Superconcentrators.
\newblock {\em {SIAM} J. Comput.}, 6(2):298--304, 1977.

\bibitem{RenD16}
Ling Ren and Srinivas Devadas.
\newblock Proof of space from stacked expanders.
\newblock In {\em Theory of Cryptography - 14th International Conference, {TCC}
  2016-B, Beijing, China, October 31 - November 3, 2016, Proceedings, Part
  {I}}, pages 262--285, 2016.

\bibitem{TCC:RenDev17}
Ling Ren and Srinivas Devadas.
\newblock Bandwidth hard functions for {ASIC} resistance.
\newblock In Yael Kalai and Leonid Reyzin, editors, {\em TCC~2017, Part~I},
  volume 10677 of {\em {LNCS}}, pages 466--492. Springer, Heidelberg, November
  2017.
\newblock \href {http://dx.doi.org/10.1007/978-3-319-70500-2_16}
  {\path{doi:10.1007/978-3-319-70500-2_16}}.

\bibitem{SS78}
John~E. Savage and Sowmitri Swamy.
\newblock Space-time trade-offs on the fft algorithm.
\newblock {\em IEEE Transactions on Information Theory}, 24(5):563--568, 1978.

\bibitem{SS79b}
John~E. Savage and Sowmitri Swamy.
\newblock Space-time tradeoffs for oblivious interger multiplications.
\newblock In {\em ICALP}, pages 498--504, 1979.

\bibitem{Svensson12}
Ola Svensson.
\newblock Hardness of vertex deletion and project scheduling.
\newblock In {\em Approximation, Randomization, and Combinatorial Optimization.
  Algorithms and Techniques - 15th International Workshop, {APPROX}, and 16th
  International Workshop, {RANDOM}. Proceedings}, pages 301--312, 2012.

\bibitem{SS79a}
Sowmitri Swamy and John~E. Savage.
\newblock Space-time tradeoffs for linear recursion.
\newblock In {\em POPL}, pages 135--142, 1979.

\bibitem{Tom78}
Martin Tompa.
\newblock Time-space tradeoffs for computing functions, using connectivity
  properties of their circuits.
\newblock In {\em Proceedings of the Tenth Annual ACM Symposium on Theory of
  Computing}, STOC '78, pages 196--204, New York, NY, USA, 1978. ACM.
\newblock URL: \url{http://doi.acm.org/10.1145/800133.804348}, \href
  {http://dx.doi.org/10.1145/800133.804348} {\path{doi:10.1145/800133.804348}}.

\end{thebibliography}
\appendix

\section{Svensson's Construction}
In this section, we review Svensson's Construction. 
Given an instance $\mathcal{U}$ of Unique Games, Svensson first constructs a weighted instance $\hat{G}_{\mathcal{U}}$ of the DAG reducibility problem. 
Recall that in the weighted DAG reducibility problem, we are given a DAG $G$ with weights on each node and a target depth $d$ and the goal is to find a minimum weight subset $S$ such that $(G-S)$ contains no path of length $d$. Given $(e_1,d_1)$ and $(e_2,d_2)$ with $e_1 < e_2$ and $d_1 > d_2$ a weaker goal is simply to distinguish between the following cases: (1) there is a set $S$ of weight at most $e_1$ s.t. $(G-S)$ contains no path of length $d_1$, and (2) for all sets $S$ of weight at most $e_2$ the graph $(G-S)$ contains a path of length $d_2$. Svensson constructs $\hat{G}_{\mathcal{U}}$ s.t. distinguishing between these cases allows us to solve the original Unique Games instance $\mathcal{U}$. 

\subsection{Notation}
We first review some notation that is used to describe Svensson's initial construction. For $x\in[k]^R$ and a subset $S$ of not necessarily distinct indices of $[R]$, let
\[C_{x,S}=\{z\in[k]^R\,:\,z_j=x_j\forall j\notin S\}\]
denote the sub-cube whose coordinates not in $S$ are fixed according to $x$. 
Let 
\[C_{x,S,v,w}=\{z\in[k]^R\,:\,z_j=x_{\pi_{v,w}(j)}\forall\pi_{v,w}(j)\notin S\}\]
denote the image of the sub-cube $C_{x,s}$ under $\pi_{v,w}$.  
Similarly, let
\[C_{x,S}^\oplus=\{z\oplus\mathbf{1}\,:\,z\in C_{x,S}\},\]
where $\oplus$ denotes addition modulo $k$ and $\mathbf{1}$ denotes an $R$-dimensional vector with all elements 1, and let
\[C_{x,S,v,w}^\oplus=\{z\oplus\mathbf{1}\,:\,z\in C_{x,S,v,w}\}.\]

\subsection{Construction} 
In Svensson's initial construction, nodes are divided into two sets: bit-vertices and test-vertices. 
Bit-vertices are assigned weight infinity to guarantee that these nodes are not deleted. 
Test-vertices are assigned weight one. 
Here, we focus on the construction of $\hat{G}_{\mathcal{U}}$ though the graph $\hat{G}_{\mathcal{U}}$ is later transformed into an instance $G_{\mathcal{U}}$ of the unweighted DAG reducibility problem, i.e., distinguishing between the cases that $G_{\mathcal{U}}$  is $(e_1,d_1)$-reducible and $(e_2,d_2)$-depth robust is sufficient to solve the original Unique Games instance $\mathcal{U}$. 
See \figref{SvenssonExample} for a simple example of $\hat{G}_{\mathcal{U}}$ along with the transformation to the unweighted instance $G_{\mathcal{U}}$. 
The DAG $\hat{G}_{\mathcal{U}}$ is defined formally as follows:

\begin{itemize}
\item
For some $L$ to be fixed, there are $L+1$ layers of bit-vertices. 
Each bit-layer $\ell$ with $0\le\ell\le L$ the DAG $\hat{G}_{\mathcal{U}}$ contains bit-vertices $b_{w,x}^\ell$ for each $w\in W$ and $x\in[k]^R$. Each bit-vertex is assigned weight $\infty$.
\item
There are $L$ layers of test-vertices. 
For each $0\le\ell\le L-1$, the DAG $\hat{G}_{\mathcal{U}}$ contains test-vertices $t_{x,S,v,w_1,\ldots,w_{2t}}^\ell$ for every $x\in[k]^R$, every sequence of indices $S=(s_1,\ldots,s_{\eps R})\in[R]^{\eps R}$, every $v\in V$ and every sequence $(w_1,\ldots,w_{2t})$ of $2t$ not necessarily distinct neighbors of $v$.  Each test-vertex is assigned weight $1$. 
\item
If $\ell\le\ell'$ and $z\in C_{x,S,v,w_j}$, then there is an edge from bit-vertex $b^\ell_{w_j,z}$ to test-vertex $t^{\ell'}_{x,S,v,w_1,\ldots,w_{2t}}$ for each $1\le j\le 2t$. 
\item
 If $\ell>\ell'$ and $z\in C^\oplus_{x,S,v,w_j}$, then there is an edge from test-vertex $t^{\ell'}_{x,S,v,w_1,\ldots,w_{2t}}$ to bit-vertex $b^\ell_{w_j,z}$ for each $1\le j\le 2t$. 
\item
If $T$ is the total number of test-vertices, then $L$ is selected so that $\delta^2 L\ge T^{1-\delta}$. 
\end{itemize}

\subsection{Transformation}
As mentioned before, in the Svensson's construction, the bit-vertices are given weight $\infty$ so that they are never deleted, and the graph can be simplified in the following manner without altering the reduction. 
The transformation to $G_{\mathcal{U}}$ is defined formally as follows:

\begin{itemize}
\item
For each $0\le\ell\le L-1$, there exists a vertex $v_{x,S,v,w_1,\ldots,w_{2t}}^\ell$ for every $x\in[k]^R$, every sequence of indices $S=(s_1,\ldots,s_{\eps R})\in[R]^{\eps R}$, every $v\in V$ and every sequence $(w_1,\ldots,w_{2t})$ of $2t$ not necessarily distinct neighbors of $v$. 
\item
If $\gamma$ is the number of vertices in each layer, then $L$ is selected so that $\delta^2 L\ge(\gamma L)^{1-\delta}$. 
\item
There exists an edge between $v_{x,S,v,w_1,\ldots,w_{2t}}^\ell$ and $v_{x',S',v',w'_1,\ldots,w'_{2t}}^{\ell'}$ if and only if $\ell<\ell'$ and there exist $i,j$ such that $C_{x,S,v,w_i}^\oplus\cap C_{x',S',v',w'_j}$ is nonempty. 
\end{itemize}


\begin{example}\exlab{toyexample}
In this example, we will illustrate how to reduce from a Unique Games instance $\mathcal{U}$ to a Svensson's construction $\hat{G}_{\mathcal{U}}$, and a simplification procedure from $\hat{G}_{\mathcal{U}}$ to $G_{\mathcal{U}}$ by examining a simple toy example. 

Consider the following Unique Games instance $\mathcal{U}=(G=(V,W,E),[R],\{\pi_{v,w}\}_{v,w})$ with $V=\{v_1\},W=\{w_1\},E=\{(v_1,w_1)\},\pi_{v_1,w_1}:\{1,2\}\rightarrow\{2,1\}$, a labeling $\rho:(V\cup W)\rightarrow [R]$ such that $\rho(v_1)=1,\rho(w_1)=2$, and with the parameters $R=2,k=2,t=1,\delta=0.1$ and $\epsilon=0.5$. Then we have the following observations when constructing $\hat{G}_{\mathcal{U}}$:

\begin{itemize}
\item Each bit-layer $\ell$ with $0\leq\ell\leq L$ contains bit-vertices $b_{w,x}^{\ell}$ for each $w\in W$ and $x\in[k]^R$. Hence, the number of bit-vertices in each layer is $|W|\times|[k]^R| = 1\times 2^2=4$. That is, for each layer $i$, we have the following bit-vertices:
\[b_{w_1,(11)}^i,b_{w_1,(12)}^i,b_{w_1,(21)}^i\text{, and }b_{w_1,(22)}^i.\]
\item Each test-layer $\ell$ with $0\leq\ell\leq L-1$ contains test-vertices $t_{x,S,v,w_1,\ldots,w_{2t}}^{\ell}$ for every $x\in[k]^R$, every sequence of indices $S=(s_1,\ldots,s_{\eps R})\in[R]^{\eps R}$, every $v\in V$ and every sequence $(w_1,\ldots,w_{2t})$ of $2t$ not necessarily distinct neighbors of $v$. Since $\eps R=1$ and $v_1$ has only one neighbor $w_1$, the number of test-vertices in each layer is $|[k]^R|\times |S|\times |V|\times |N_G(v_1)^{2t}|=2^2\times 2\times 1\times 1^2=8$, where $N_G(v)$ denotes the set of neighbors of $v$ in a graph $G$. That is, for each layer $i$, we have the following test-vertices (from now on, we omit the subscript $v,w_1,\ldots,w_{2t}$ in this example since there is only one such case for each test-vertex):
\[t_{(11),(1)}^i,t_{(12),(1)}^i,t_{(21),(1)}^i,t_{(22),(1)}^i,t_{(11),(2)}^i,t_{(12),(2)}^i,t_{(21),(2)}^i\text{, and }t_{(22),(2)}^i.\]
\item There exists an edge from bit vertex $b_{w_1,z}^{\ell}$ to test-vertex $t_{x,S}^{\ell'}$ if $\ell\leq\ell'$ and $z\in C_{x,S,v_1,w_1}$. We recall that $C_{x,S,v_1,w_1}=\{z\in[k]^R:z_j=x_{\pi_{v_1,w_1}(j)}~\forall\pi_{v_1,w_1}(j)\not\in S\}$. Now it is easy to see that if $S=\{1\}$, $z\in C_{x,S,v_1,w_1}$ if and only if $z_1=x_2$, and if $S=\{2\}$, $z\in C_{x,S,v_1,w_1}$ if and only if $z_2=x_1$. Therefore, we have an edge from $b_{w_1,(12)}^i$ to $t_{(11),(1)}^j,t_{(21),(1)}^j,t_{(21),(2)}^j$, and $t_{(22),(2)}^j$ for all $0\leq i\leq j<L$, and so on.
\item There exists an edge from test-vertex $t_{x,S}^{\ell'}$ to bit-vertex $b_{w_1,z}^{\ell}$ if $\ell>\ell'$ and $z\in C_{x,S,v_1,w_1}^{\oplus}$, where $\oplus$ denotes addition modulo $k=2$ and $C_{x,S,v_1,w_1}^{\oplus} = \{z\oplus 1:z\in C_{x,S,v_1,w_1}\}$. Hence, for example, if there is an edge from $b_{w_1,(12)}^i$ to $t_{x,S}^j$ then there should be an edge from $t_{x,S}^j$ to $b_{w_1,(21)}^{j'}$ for all $j'>j$ since $(12)\oplus \mathbf{1} = (12)\oplus(11)=(21)$.
\end{itemize}

\noindent When transforming $\hat{G}_{\mathcal{U}}$ into $G_{\mathcal{U}}$, we can observe that $C_{x,S,v_1,w_1}^{\oplus} \cap C_{x',S',v_1,w_1}$ is nonempty if and only if there is a path between two test-vertices through one bit-vertex. Taken together, we have the following structure of graphs reduced from a Unique Games instance $\mathcal{U}$, as shown in \figref{SvenssonExample}.\hfill\textcolor{darkgray}{$\blacktriangleleft$}

\begin{figure}
\begin{tikzpicture}
    [node distance=1cm,auto,font=\small,
    every node/.style={node distance=2cm},
    bit/.style={rectangle,draw,fill=purdueevertrueblue!20,text=black,inner sep=2pt},
    test/.style={rectangle,draw,fill=purduecampusgold!20,text=black,inner sep=2pt},
    bitblue/.style={rectangle,draw,blue,fill=purdueevertrueblue!20,text=blue,inner sep=2pt},
    bitgreen/.style={rectangle,draw,tangocolordarkchameleon,fill=purdueevertrueblue!20,text=tangocolordarkchameleon,inner sep=2pt},
    testred/.style={rectangle,draw,red,fill=purduecampusgold!20,text=red,inner sep=2pt},
    testmagenta/.style={rectangle,draw,magenta,fill=purduecampusgold!20,text=magenta,inner sep=2pt},
    circ/.style ={ circle ,top color = purduecoalgray!50 , bottom color = purduecoalgray,
    draw,purduecoalgray, text=white ,inner sep=-0.5pt, minimum width =0.6 cm}]
    ]
    \tikzset{decoration={snake,amplitude=.4mm,segment length=2mm,
                       post length=1.5mm,pre length=1mm}}
                       
	\node[anchor=west] (UGinstance) at (0,0) {Unique Games instance $\mathcal{U}$:};
	\node[circ] (v1) at (6,-1) {$v_1$};
	\node[circ] (w1) at (9,-1) {$w_1$};
	\path[draw,->] (v1) -- (w1);
	\node[below=1pt of v1] (V) {$V$}; 
	\node[below=1pt of w1] (W) {$W$}; 
	\node[anchor=west] (R) at ($(w1) + (1,0.3)$) {$[R]:=\{1,2\}$};
	\node[anchor=west] (pi) at ($(w1) + (1,-0.3)$) {$\pi_{v_1,w_1}:\{1,2\}\rightarrow\{2,1\}$};
	
	\node[anchor=west] (Reduction) at (0,-2) {Reduction $\hat{G}_{\mathcal{U}}$:};
	\node (vd) at (8.25,-2.3) {$\vdots$};
	\node (vd) at (8.25,-9.3) {$\vdots$};
	\node (vd) at (8.25,-10.3) {$\vdots$};
	\node (vd) at (8.25,-13.5) {$\vdots$};
	\node (Ti+1) at ($(Reduction.west) + (1.5,-1)$) {$T_{i+1}$};
	\node (vd) at ($(Ti+1) + (0,0.6)$) {$\vdots$};
	\node[test] (test_i+1_1) at ($(Ti+1) + (1.5,0)$) {$t_{(11),(1)}^{i+1}$};
	\node[test] (test_i+1_2) at ($(Ti+1) + (3,0)$) {$t_{(12),(1)}^{i+1}$};
	\node[test] (test_i+1_3) at ($(Ti+1) + (4.5,0)$) {$t_{(21),(1)}^{i+1}$};
	\node[test] (test_i+1_4) at ($(Ti+1) + (6,0)$) {$t_{(22),(1)}^{i+1}$};
	\node[test] (test_i+1_5) at ($(Ti+1) + (7.5,0)$) {$t_{(11),(2)}^{i+1}$};
	\node[test] (test_i+1_6) at ($(Ti+1) + (9,0)$) {$t_{(12),(2)}^{i+1}$};
	\node[test] (test_i+1_7) at ($(Ti+1) + (10.5,0)$) {$t_{(21),(2)}^{i+1}$};
	\node[test] (test_i+1_8) at ($(Ti+1) + (12,0)$) {$t_{(22),(2)}^{i+1}$};
	\node (Bi+1) at ($(Ti+1) + (0,-2)$) {$B_{i+1}$};
	\node[bit] (bit_i+1_1) at ($(Bi+1) + (2.25,0)$) {$b_{w_1,(11)}^{i+1}$};
	\node[bitblue] (bit_i+1_2) at ($(Bi+1) + (5.25,0)$) {$b_{w_1,(12)}^{i+1}$};
	\node[bit] (bit_i+1_3) at ($(Bi+1) + (8.25,0)$) {$b_{w_1,(21)}^{i+1}$};
	\node[bitgreen] (bit_i+1_4) at ($(Bi+1) + (11.25,0)$) {$b_{w_1,(22)}^{i+1}$};
	\node (Ti) at ($(Bi+1) + (0,-2)$) {$T_i$};
	\node[test] (test_i_1) at ($(Ti) + (1.5,0)$) {$t_{(11),(1)}^i$};
	\node[test] (test_i_2) at ($(Ti) + (3,0)$) {$t_{(12),(1)}^i$};
	\node[test] (test_i_3) at ($(Ti) + (4.5,0)$) {$t_{(21),(1)}^i$};
	\node[test] (test_i_4) at ($(Ti) + (6,0)$) {$t_{(22),(1)}^i$};
	\node[testred] (test_i_5) at ($(Ti) + (7.5,0)$) {$t_{(11),(2)}^i$};
	\node[test] (test_i_6) at ($(Ti) + (9,0)$) {$t_{(12),(2)}^i$};
	\node[test] (test_i_7) at ($(Ti) + (10.5,0)$) {$t_{(21),(2)}^i$};
	\node[test] (test_i_8) at ($(Ti) + (12,0)$) {$t_{(22),(2)}^i$};
	\node (Bi) at ($(Ti) + (0,-2)$) {$B_i$};
	\node (vd) at ($(Bi) + (0,-0.4)$) {$\vdots$};
	\node[bit] (bit_i_1) at ($(Bi) + (2.25,0)$) {$b_{w_1,(11)}^i$};
	\node[bitblue] (bit_i_2) at ($(Bi) + (5.25,0)$) {$b_{w_1,(12)}^i$};
	\node[bit] (bit_i_3) at ($(Bi) + (8.25,0)$) {$b_{w_1,(21)}^i$};
	\node[bitgreen] (bit_i_4) at ($(Bi) + (11.25,0)$) {$b_{w_1,(22)}^i$};
	\path[decorate,draw,->,red] (test_i_5.north) -- (bit_i+1_2.south);
	\path[decorate,draw,->,red] (test_i_5.north) -- (bit_i+1_4.south);
	
	\path[decorate,draw,->,blue] (bit_i+1_2.north) -- (test_i+1_1.south);
	\path[decorate,draw,->,blue] (bit_i+1_2.north) -- (test_i+1_3.south);
	\path[decorate,draw,->,blue] (bit_i+1_2.north) -- (test_i+1_7.south);
	\path[decorate,draw,->,blue] (bit_i+1_2.north) -- (test_i+1_8.south);
	
	\path[decorate,draw,->,tangocolordarkchameleon] (bit_i+1_4.north) -- (test_i+1_2.south);
	\path[decorate,draw,->,tangocolordarkchameleon] (bit_i+1_4.north) -- (test_i+1_4.south);
	\path[decorate,draw,->,tangocolordarkchameleon] (bit_i+1_4.north) -- (test_i+1_7.south);
	\path[decorate,draw,->,tangocolordarkchameleon] (bit_i+1_4.north) -- (test_i+1_8.south);
	
	\draw[blue,->,dashed] (bit_i_2.north) to[bend left=45] (test_i+1_1.south);
	\draw[blue,->,dashed] (bit_i_2.north) to[bend left=15] (test_i+1_3.south);
	\draw[blue,->,dashed] (bit_i_2.north) to[bend right=25] (test_i+1_7.south);
	\draw[blue,->,dashed] (bit_i_2.north) to[bend right=35] (test_i+1_8.south);
	
	\path[draw,->,blue] (bit_i_2.north) -- (test_i_1.south);
	\path[draw,->,blue] (bit_i_2.north) -- (test_i_3.south);
	\path[draw,->,blue] (bit_i_2.north) -- (test_i_7.south);
	\path[draw,->,blue] (bit_i_2.north) -- (test_i_8.south);
	
	\path[draw,->,tangocolordarkchameleon] (bit_i_4.north) -- (test_i_2.south);
	\path[draw,->,tangocolordarkchameleon] (bit_i_4.north) -- (test_i_4.south);
	\path[draw,->,tangocolordarkchameleon] (bit_i_4.north) -- (test_i_7.south);
	\path[draw,->,tangocolordarkchameleon] (bit_i_4.north) -- (test_i_8.south);
	
	\draw[tangocolordarkchameleon,->,dashed] (bit_i_4.north) to[bend left=35] (test_i+1_2.south);
	\draw[tangocolordarkchameleon,->,dashed] (bit_i_4.north) to[bend left=15] (test_i+1_4.south);
	\draw[tangocolordarkchameleon,->,dashed] (bit_i_4.north) to[bend right=5] (test_i+1_7.south);
	\draw[tangocolordarkchameleon,->,dashed] (bit_i_4.north) to[bend right=15] (test_i+1_8.south);
	
	\node at ($(bit_i+1_1) + (0,1)$) {$\cdots$};
	\node at ($(bit_i+1_1) + (4,1)$) {$\cdots$};
	\node at ($(bit_i+1_1) + (8,1)$) {$\cdots$};
	\node at ($(bit_i+1_1) + (0,-1)$) {$\cdots$};
	\node at ($(bit_i+1_1) + (4,-1)$) {$\cdots$};
	\node at ($(bit_i+1_1) + (8,-1)$) {$\cdots$};
	\node at ($(bit_i+1_1) + (0,-3)$) {$\cdots$};
	\node at ($(bit_i+1_1) + (4,-3)$) {$\cdots$};
	\node at ($(bit_i+1_1) + (8,-3)$) {$\cdots$};
	
	\node[anchor=west] (Transformation) at (0,-10) {Transformation $G_{\mathcal{U}}$:};
	\node (Vi+1) at ($(Transformation.west) + (1.5,-1)$) {\footnotesize{Layer $(i+1)$}};
	\node (vd) at ($(Vi+1) + (0,0.6)$) {$\vdots$};
	\node[test] (v_i+1_1) at ($(Vi+1) + (1.5,0)$) {$v_{(11),(1)}^{i+1}$};
	\node[test] (v_i+1_2) at ($(Vi+1) + (3,0)$) {$v_{(12),(1)}^{i+1}$};
	\node[test] (v_i+1_3) at ($(Vi+1) + (4.5,0)$) {$v_{(21),(1)}^{i+1}$};
	\node[test] (v_i+1_4) at ($(Vi+1) + (6,0)$) {$v_{(22),(1)}^{i+1}$};
	\node[test] (v_i+1_5) at ($(Vi+1) + (7.5,0)$) {$v_{(11),(2)}^{i+1}$};
	\node[test] (v_i+1_6) at ($(Vi+1) + (9,0)$) {$v_{(12),(2)}^{i+1}$};
	\node[test] (v_i+1_7) at ($(Vi+1) + (10.5,0)$) {$v_{(21),(2)}^{i+1}$};
	\node[test] (v_i+1_8) at ($(Vi+1) + (12,0)$) {$v_{(22),(2)}^{i+1}$};
	\node (Vi) at ($(Vi+1) + (0,-2)$) {\footnotesize{Layer $i$}};
	\node (vd) at ($(Vi) + (0,-0.4)$) {$\vdots$};
	\node[test] (v_i_1) at ($(Vi) + (1.5,0)$) {$v_{(11),(1)}^i$};
	\node[test] (v_i_2) at ($(Vi) + (3,0)$) {$v_{(12),(1)}^i$};
	\node[test] (v_i_3) at ($(Vi) + (4.5,0)$) {$v_{(21),(1)}^i$};
	\node[test] (v_i_4) at ($(Vi) + (6,0)$) {$v_{(22),(1)}^i$};
	\node[testmagenta] (v_i_5) at ($(Vi) + (7.5,0)$) {$v_{(11),(2)}^i$};
	\node[test] (v_i_6) at ($(Vi) + (9,0)$) {$v_{(12),(2)}^i$};
	\node[test] (v_i_7) at ($(Vi) + (10.5,0)$) {$v_{(21),(2)}^i$};
	\node[test] (v_i_8) at ($(Vi) + (12,0)$) {$v_{(22),(2)}^i$};
	\path[decorate,draw,->,magenta] (v_i_5.north) -- (v_i+1_1.south);
	\path[decorate,draw,->,magenta] (v_i_5.north) -- (v_i+1_3.south);
	\path[decorate,draw,->,magenta] (v_i_5.north) -- (v_i+1_7.south);
	\path[decorate,draw,->,magenta] (v_i_5.north) -- (v_i+1_8.south);
	\path[decorate,draw,->,magenta] (v_i_5.north) -- (v_i+1_2.south);
	\path[decorate,draw,->,magenta] (v_i_5.north) -- (v_i+1_4.south);
	
	\node at ($(v_i_1) + (0.75,1)$) {$\cdots$};
	\node at ($(v_i_1) + (6,1)$) {$\cdots$};
	\node at ($(v_i_1) + (9.75,1)$) {$\cdots$};
\end{tikzpicture}
\caption{An example of reduction $\hat{G}_{\mathcal{U}}$ from a Unique Games instance $\mathcal{U}$ and a transformation into $G_{\mathcal{U}}$ illutrated in \exref{toyexample}. We remark that in $\hat{G}_{\mathcal{U}}$ and $G_{\mathcal{U}}$, we only drew edges from the highlighted vertices for simplicity and readability. In $G_{\mathcal{U}}$, we only keep test-vertices from $\hat{G}_{\mathcal{U}}$ and connect two vertices if there is a path between those two through one bit-vertex. Therefore, we can easily check that totally 6 edges (shown in snaked magenta edges) going out from the vertex $v_{(11),(2)}^i$ in $G_{\mathcal{U}}$ comes from the edges $(t_{(11),(2)}^i,b_{w_1,(12)}^{i+1})$ and $(t_{(11),(2)}^i,b_{w_1,(22)}^{i+1})$ (shown in a snaked red edge) and edges from $b_{w_1,(12)}^{i+1}$ and $b_{w_1,(22)}^{i+1}$ to the test layer $T_{i+1}$ (shown in snaked blue/green edges) in $\hat{G}_{\mathcal{U}}$. We also note that the edges starting from the $i\th$ bit layer $B_i$ go to every upper test layers $T_j$ for all $j\geq i$. Those edges are represented as dashed ones.}
\figlab{SvenssonExample}
\end{figure}

\end{example}

\section{Modified Construction}
\applab{app:modified}
Given an instance $\hat{G}_{\mathcal{U}}$ of the Svensson's construction and a $\gamma$-extreme depth-robust graph $G_{\gamma,L+1} = (V_\gamma = [L+1],E_\gamma)$, we formally define our modified instance $G'=\bitsparsify_{G_{\gamma,L+1}}(\hat{G}_{\mathcal{U}})$ in the following manner.
\begin{mdframed}
\begin{center}
Transformation $\bitsparsify_{G_{\gamma,L+1}}(\hat{G}_{\mathcal{U}})$
\end{center}
\noindent
\underline{Input}: An instance $\hat{G}_{\mathcal{U}}=(V,E)$ of the Svensson's construction, whose vertices are partitioned into $L+1$ bit-layers $B_0,\ldots,B_{L}$ and $L$ test-layers $T_0,\ldots,T_{L-1}$, a $\gamma$-extreme depth robust graph $G_{\gamma,L+1} = (V_\gamma = [L+1],E_\gamma)$. 

\begin{enumerate}[1.]
\item
Let $G'=(V,E)$ be a copy of $\hat{G}_{\mathcal{U}}$. 
\item
If $e=(b,t)$ is an edge in $\hat{G}_{\mathcal{U}}$, where $b\in B_{i}$ and $t\in T_{j}$, delete $e$ from $G'$ if $i \neq j$ and $(i,j) \not \in E_\gamma$.
\item
If $e=(t,b)$ is an edge in $\hat{G}_{\mathcal{U}}$, where $b\in B_{i}$ and $t\in T_{j}$, delete $e$ from $G'$ if $(j,i) \not \in E_\gamma$.
\end{enumerate}

\noindent
\underline{Output}: $G'$
\end{mdframed}
We give an illustration of the $\mathsf{Sparsify}$ procedure in \figref{fig:sparsify:proc}.
\begin{figure}[h]
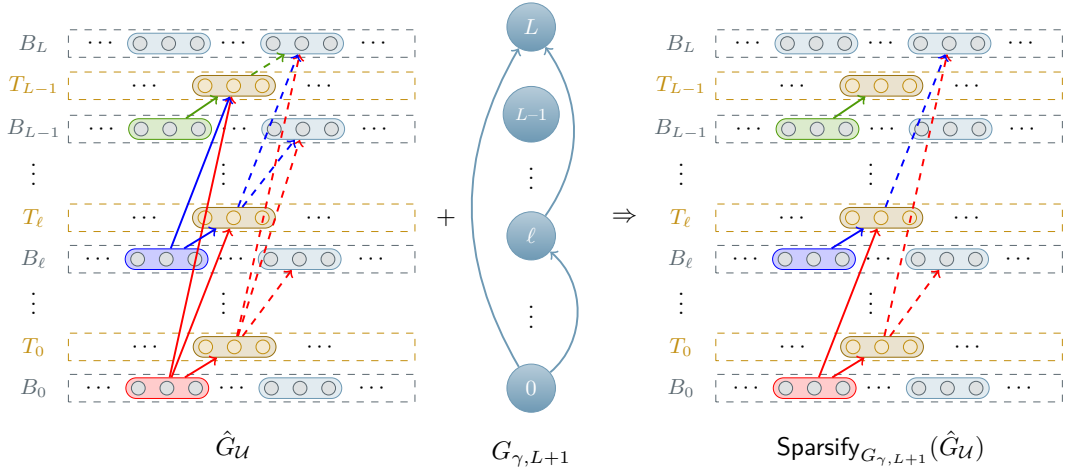

\centering
\includestandalone[width=\textwidth]{figure-indegree}
\caption{A description of the transformation $\mathsf{Sparsify}_{G_{\gamma,L+1}}(\hat{G}_{\mathcal{U}})$ where $\hat{G}_{\mathcal{U}}$ is Svensson's construction and $G_{\gamma,L+1}$ is a $\gamma$-extreme depth-robust graph. We remark that the edges between the subsets of nodes indicate that every node in the input subset is connected to every node in the output subset (in this example, there should be $3\times 3 = 9$ edges from $B_0$ to $T_0$.)}
\figlab{fig:sparsify:proc}
\end{figure}

Correspondingly, our modified instance can also be simplified in the following manner without altering the reduction.

\begin{itemize}
\item
For a input parameter $\gamma$, let $G_{\gamma,L+1} = (V_\gamma = [L+1],E_\gamma)$ be an $\frac{\gamma}{2}$-extreme depth robust graph with $L+1$ vertices, which we use $[L+1]$ to represent.  
\item
For each $0\le\ell\le L-1$, there exists a vertex $v_{x,S,v,w_1,\ldots,w_{2t}}^\ell$ for every $x\in[k]^R$, every sequence of indices $S=(s_1,\ldots,s_{\eps R})\in[R]^{\eps R}$, every $v\in V$ and every sequence $(w_1,\ldots,w_{2t})$ of $2t$ not necessarily distinct neighbors of $v$. 
\item
If $\gamma$ is the number of vertices in each layer, then $L$ is selected so that $\delta^2 L\ge(\gamma L)^{1-\delta}$. 
\item
There exists an edge between $v_{x,S,v,w_1,\ldots,w_{2t}}^\ell$ and $v_{x',S',v',w'_1,\ldots,w'_{2t}}^{\ell'}$ if and only if $\ell<\ell'$, the edge $(\ell,\ell')$ is in $E_\gamma$, and there exist $i,j$ such that $C_{x,S,v,w_i}^\oplus\cap C_{x',S',v',w'_j}$ is nonempty. 
\end{itemize}


We first recall the following definition of influence of the $i\th$ coordinate:
\[\Infl_i(f)=\mathbb{E}_x\left[\var(f)|x_1,\ldots,x_{i-1},x_{i+1},\ldots,x_R\right].\]
We now reference the key theorem used in Svensson's analysis.
\begin{theorem}\cite{STOC:Khot02b,Svensson12}
\thmlab{thm:aintover}
For every $\eps,\delta>0$ and integer $k$, there exists $\eta>0$ and integers $t,d$ such that any collection of functions $f_1,\ldots,f_t:[k]^R\rightarrow\{0,1\}$ that satisfies
\begin{itemize}
\item
$\forall j,\mathbb{E}\left[f_j\right]\ge\delta$
\item
$\forall i\in[R]$, $\forall 1\le\ell_1\neq\ell_2\le t$: $\min\left\{\Infl_i^d(f_{\ell_1}),\Infl_i^d(f_{\ell_2})\right\}\le\eta$
\end{itemize}
has
\[\underset{x,S_\eps}{\Pr}\left[\bigwedge_{j=1}^t f_j(C_{x,S_{\eps}})\equiv0\right]\le\delta.\]
\end{theorem}

We now show that the transformed graph maintains similar properties as Svensson's construction, given an instance of Unique Games. 
The following statement is analogous to Lemma 4.7 in \cite{Svensson12}.
%

\begin{remindertheorem}{\lemref{lem:svensson:exist}}
\lemsvenssonexist
\end{remindertheorem}
\begin{proof}[Proof of \lemref{lem:svensson:exist}]
As in \cite{Svensson12}, an equivalent formulation of the problem is finding a coloring $\chi$ in $\{1,2,\ldots,k\}$ to each bit-vertex to minimize the number of unsatisfied test-vertices. 
Unlike in \cite{Svensson12}, we say a test-vertex $t^\ell_{x,S,v,w_1,\ldots,w_{2t}}$ is satisfied if
\[\underset{\substack{1\le j\le 2t\\z\in C_{x,S,v,w_j}}}{\max}\chi\left(b^{\ell'}_{w_j,z}\right)<\underset{\substack{1\le j\le 2t\\z\in C^\oplus_{x,S,v,w_j}}}{\min}\chi\left(b^{\ell''}_{w_j,z}\right),\]
for all $\ell'\le\ell<\ell''$ with $(\ell',\ell'')\in E_{\gamma}$, so that all the predecessors of $\ell$ are assigned lower colors than the successors of $\ell$.  

We also define the color $\chi(w,\ell)$ for $w\in W$ and $0\le\ell\le L$ as the maximum color that satisfies
\[\Pr_x\left[\chi(b^\ell_{w,x})\ge\chi(w,\ell)\right]\ge1-\delta.\]
For $w\in W$ and $0\le\ell\le L-1$, define the indicator function $f_w^\ell:[k]^R\rightarrow\{0,1\}$ by
\[f_w^\ell(x)=
\begin{cases}
0\qquad\text{if }\chi\left(b^{\ell'}_{w,x}\right)>\chi(w,\ell)\text{ for all }\ell'>\ell\text{ with }(\ell,\ell')\in E_{\gamma},\\
1\qquad\text{otherwise}.
\end{cases}
\]
We call a test-vertex $w\in W$ good in test-layer $\ell$ if for $\chi(w,\ell') > \chi(w,\ell)$ for every edge of the form $(\ell,\ell')$ in the $\gamma$-extreme depth-robust graph $G_{\gamma,L+1}=(V_\gamma = [L+1],E_\gamma)$. 
\begin{claim}
\claimlab{clm:ug}
If the Unique Games instance has no labeling satisfying a fraction $\frac{\delta\eta^2}{t^2k^2}$ of the constraints and a fraction $16\delta$ of the vertices of test-layer $\ell$ are satisfied, then at least a $2\delta$ fraction of the vertices are good in test-layer $\ell$.
\end{claim}
\begin{claimproof}
Let $A_{\ell}$ be the set of satisfied vertices of test-layer $\ell$ so that for all $\ell'\le\ell<\ell''$ with $(\ell',\ell'')\in E_\gamma$, it follows that
\[\underset{x,S,v,w_1,\ldots,w_{2t}}{\Pr}\left[\underset{\substack{1\le j\le 2t\\z\in C_{x,S,v,w_j}}}{\max}\chi\left(b^{\ell'}_{w_j,z}\right)<\underset{\substack{1\le j\le 2t\\z\in C^\oplus_{x,S,v,w_j}}}{\min}\chi\left(b^{\ell''}_{w_j,z}\right)\right]\ge16\delta,\]
since at least $16\delta$ fraction of the vertices in $A_{\ell}$ are satisfied. 
We call a tuple $(v,w_1,\ldots,w_{2t})$ good if
\[\underset{x,S}{\Pr}\left[\underset{\substack{1\le j\le 2t\\z\in C_{x,S,v,w_j}}}{\max}\chi\left(b^{\ell'}_{w_j,z}\right)<\underset{\substack{1\le j\le 2t\\z\in C^\oplus_{x,S,v,w_j}}}{\min}\chi\left(b^{\ell''}_{w_j,z}\right)\right]\ge8\delta,\]
for all $\ell'\le\ell<\ell''$ with $(\ell',\ell'')\in E_\gamma$. 
Observe that at least $8\delta$ fraction of the tuples are good. 

From the definition of $\chi(w,\ell')$, we have that $\underset{x}{\Pr}\left[\chi\left(b^{\ell'}_{w,x}\right)\right]\ge1-\delta$. 
Hence for a good tuple, it follows that
\[7\delta\le\underset{x,S}{\Pr}\left[\underset{1\le j\le 2t}{\max}\chi(w_j,\ell')<\underset{\substack{1\le j\le 2t\\z\in C^\oplus_{x,S,v,w_j}}}{\min}\chi\left(b^{\ell''}_{w_j,z}\right)\right]\le\underset{x,S}{\Pr}\left[\bigwedge_{j=1}^{2t}f^\ell_{w_j}(C_{x,S,v,w_j})\equiv0\right],\]
for all $\ell'\le\ell<\ell''$ with $(\ell',\ell'')\in E_\gamma$. 

Therefore by \thmref{thm:aintover}, at least one of the following cases holds:
\begin{enumerate}[1.]
\item
\label{cond:goodone}
more than $t$ of the functions have $\mathbb{E}\left[f_{w_j}^\ell\right]<\delta$ so that $\chi(w_j,\ell')>\chi(w_j,\ell)$ for every edge of the form $(\ell,\ell')$ in $E_\gamma$, or
\item
\label{cond:goodtwo}
there exist $1\le\ell_1\neq\ell_2\le t$ and $j\in\mathcal{I}\left[w_{\ell_1}\right], j'\in\mathcal{I}\left[w_{\ell_2}\right]$ such that $\pi_{v,w_{\ell_1}}(j)=\pi_{v,w_{\ell_2}}(j')$, where
\[\mathcal{I}[w]=\{i\in[R]\,:\,\Infl_i^d(f_w^\ell)\ge\eta\}.\]
\end{enumerate}
If Condition~\ref{cond:goodone} holds for at least $\frac{1}{2}$ of the good tuples, or equivalently a $4\delta$ fraction of all tuples, then at least a $2\delta$ fraction of the test-vertices are good in test-layer $\ell$ because we can pick a vertex $w_j\in W$ uniformly at random by picking a tuple $(v,w_1,\ldots,w_{2t})$ and then taking one of the vertices $w_1,\ldots,w_{2t}$ at random. 
Conditioned on the (at least) $2\delta$ probability that the tuple is good and satisfies Condition~\ref{cond:goodone}, the probability that $\chi(w_j,\ell')>\chi(w_j,\ell)$ for every edge of the form $(\ell,\ell')$ in $E_\gamma$ for the sampled vertex $w_j$ is at least $\frac{1}{2}$. 
Therefore, $\underset{w}{\Pr}\left[\chi(w,\ell')>\chi(w,\ell)\text{ for all }\ell'>\ell\text{ with }(\ell,\ell')\in E_\gamma\right]\ge2\delta$, so that at least $2\delta$ fraction of the test-vertices are good in test-layer $\ell$.

By way of contradiction, we can show that if Condition~\ref{cond:goodtwo} were to hold for more than half of the good tuples, the assumption that the Unique Games instance has no labeling satisfying a fraction $\frac{\delta\eta^2}{t^2k^2}$ of the constraints is violated. 
The argument holds exactly as Claim 4.8 in \cite{Svensson12}, but we repeat it here for completeness.

For every $w\in W$, let $\rho(w)$ be a random label from $\mathcal{I}[w]$. 
For every $v\in V$, let $w$ be a random neighbor of $v$ and let $\rho(v)=\pi_{v,w}(\rho(w))$. 
If Condition~\ref{cond:goodtwo} holds for half of the good tuples, then a random tuple has this property with at least probability $4\delta$. 
Thus with probability at least $\frac{1}{4t^2}$, $w=w_{\ell_1}$ and $w'=w_{\ell_2}$ for $w,w'$ randomly picked from the set $\{w_1,\ldots,w_{2t}\}$. 
Moreover, \cite{Svensson12,BansalK09} observes that with probability at least $\frac{\eta^2}{k^2}$, the labeling procedure defines $j=\rho(w)$ and $j'=\rho(w')$. 
Hence if Condition~\ref{cond:goodtwo} holds for half of the good tuples,
\[\underset{v,w,w'}{\Pr}\left[\pi_{v,w}(\rho(w))=\pi_{v,w'}(\rho(w'))\right]\ge\frac{4\delta\eta^2}{4t^2k^2},\]
so that over the randomness of the labeling procedure,
\[\underset{(v,w)}{\Pr}\left[\rho(v)=\pi_{v,w}(\rho(w))\right]\ge\frac{\delta\eta^2}{t^2k^2},\]
which contradicts the assumption that the Unique Games instance has no labeling satisfying a fraction $\frac{\delta\eta^2}{t^2k^2}$ of the constraints is violated. 
\end{claimproof}

Consider a subgraph induced by all bit-vertices and a fraction $32\delta$ of the test-vertices and consider the minimum number of colors required for a coloring $\chi$ to satisfy the $32\delta$ fraction of the test-vertices. 
Since at least $16\delta$ fraction of the test-vertices are good in at least $16\delta$ fraction of the test-layers, then by \claimref{clm:ug},
\[\underset{\ell\in[L],w\in W}{\Pr}\left[\chi(w,\ell')>\chi(w,\ell)\text{ for all }\ell'>\ell\text{ with }(\ell,\ell')\in E_\gamma\right]\ge16\delta\cdot2\delta=32\delta^2.\]
Therefore, there exists $w\in W$ with $\underset{\ell\in[L]}{\Pr}\left[\chi(w,\ell')>\chi(w,\ell)\text{ for all }\ell'>\ell\text{ with }(\ell,\ell')\in E_\gamma\right]\ge32\delta^2$.
\end{proof}

Interestingly, we can exactly compute the pebbling complexity of the simplified Svensson's construction, when the graph is only represented with the test-vertices.
\begin{lemma} \lemlab{PebblingCostSvensson}
Given a (simplified) Svensson's construction $G_{\mathcal{U}}$ that consists of $N$ vertices partitioned across $L$ layers, $\pcc(G_{\mathcal{U}})=\frac{N(L+1)}{2}$. 
\end{lemma}
\begin{proof}
We first show that the pebbling complexity of $G_{\mathcal{U}}$ is at least $\frac{N(L+1)}{2}$. 
Observe that the Svensson's construction contains $\frac{N}{L}$ vertices in each layer and furthermore, for each pair of layers $i$ and $j$, there is a perfect matching between vertices of layer $i$ and vertices of layer $j$ among the edges connecting layers $i$ and $j$.  
Let $\mathcal{M}_{i,j}$ be the subset of edges between $i$ and $j$ that is perfect matching. 

For a given pebble $u$ in layer $j$, let $v_i$ be the vertex in layer $i$ matched to $u$ by $\mathcal{M}_{i,j}$. 
To pebble a vertex $u$ in layer $j$, all of its parents must contain pebbles in the previous round. 
Namely, there must be a pebble on $v_i$ for all $1\le i<j$ in the previous round. 
Since each $\mathcal{M}_{i,j}$ is a perfect matching, there must be $j-1$ pebbles on the graph solely for the purpose of pebbling node $u$ in layer $j$. 
Thus pebbling each node in layer $j$ induces a pebbling cost of at least $j-1$. 
Since there are $\frac{N}{L}$ pebbles in each layer and $L$ layers, then the total pebbling cost is at least $\frac{N}{L}\sum_{j=1}^L(j-1)=\frac{N(L+1)}{2}$, which lower bounds $\pcc(G_{\mathcal{U}})$. 

On the other hand, consider the natural pebbling where all the pebbles in layer $j$ are pebbled in round $j$, and no pebble is ever removed. 
Then the graph is completely pebbled in $L$ rounds, since layer $L$ is pebbled in round $L$. 
Moreover, the cost of pebbling round $j$ is $\frac{N}{L}(j-1)$. 
Hence, the pebbling cost is $\frac{N}{L}\sum_{j=1}^L(j-1)=\frac{N(L+1)}{2}$, which upper bounds $\pcc(G_{\mathcal{U}})$. 
\end{proof}

Finally, we give a formal description of the procedure $\idr(G,\gamma)$. 
Recall that $\idr(G,\gamma)$ for a graph $G=(V,E)$ replaces each vertex $v\in V$ with a path $P_v = v_1,\ldots,v_{\delta+\gamma}$, where $\delta$ is the indegree of $G$. 
For each edge $(u,v) \in E$, we add the edge $(u_{\delta+\gamma},v_i)$ whenever $(u,v)$ is the $i\th$ incoming edge of $v$, according to some fixed ordering. 
\cite{EC:AlwBloPie17} give parameters $e_2$ and $d_2$ so that $\idr(G,\gamma)$ is $(e_2,d_2)$-depth robust if $G$ is $(e,d)$-depth robust. 
We complete the reduction by giving parameters $e_1$ and $d_1$ so that $\idr(G,\gamma)$ is $(e_1,d_1)$-reducible if $G$ is $(e,d)$-reducible. 

\begin{mdframed}
\begin{center}
Transformation $\idr(G,\gamma)$
\end{center}
\noindent
\underline{Input}: An DAG $G=(V,E)$ with indegree $\delta$, parameter $\gamma$. 

\begin{enumerate}[1.]
\item
Let the vertices of $G$ be $\left[|V|\right]$.
\item
Initialize $G'$ to be a graph with $(\delta+\gamma)|V|$ vertices and let these vertices be $\left[(\delta+\gamma)|V|\right]$
\item
If $(\delta+\gamma)n+1\le u<(\delta+\gamma)n$ for some integer $n$, add edge $(u,u+1)$ to $G'$. 
\item
If $(u,v)$ is the $i\th$ incoming edge of $v$ by some fixed predetermined ordering, then add $(u_{\delta+\gamma},v_i)$ to $G'$.
\end{enumerate}

\noindent
\underline{Output}: $G'$
\end{mdframed}

We given an illustration of the $\idr$ transformation in \figref{fig:idr}.

\begin{figure}[h]
\centering
\begin{tikzpicture}
    [node distance=1cm,auto,font=\small,
    every node/.style={node distance=2cm},
    bit/.style={circle,draw,purdueevertrueblue,fill=purdueevertrueblue!20,text=purdueevertrueblue,inner sep=2pt},
    test/.style={circle,draw,purduecampusgold,fill=purduecampusgold!20,text=purduecampusgold,inner sep=2pt},
    circ/.style ={ circle ,top color = purduecoalgray!50 , bottom color = purduecoalgray,
    draw,purduecoalgray, text=white ,inner sep=-0.5pt, minimum width =0.6 cm}]
    ]
    \node[circ] (v) {$v$};
    \node[coordinate,position=140:{0.6cm} from v] (a) {};
    \node[circ,position=110:{0.6cm} from v] (u) {$u$};
    \node[rotate=-10,position=80:{0.3cm} from v] (cd1) {$\cdots$};
    \node[coordinate,position=40:{0.6cm} from v] (b) {};
    \node[below=0.2cm of v] (original) {\large $G$};
    \path[draw,color=blue,->] (a) -- (v);
    \path[draw,color=red,->] (u) -- (v);
    \path[draw,color=tangocolordarkchameleon,->] (v) -- (b);
    \node[circ,right=5cm of v] (v1) {$v_1$};
    \node[circ,right=0.4cm of v1] (v2) {$v_2$};
    \node[right=0.4cm of v2] (cd2) {$\cdots$};
    \node[circ,right=0.4cm of cd2] (vN) {$v_{{}_{\delta+\gamma}}$};
    \node[circ,right=3.5cm of u] (u1) {$u_1$};
    \node[right=0.4cm of u1] (cd3) {$\cdots$};
    \node[circ,right=0.4cm of cd3] (uN) {$u_{{}_{\delta+\gamma}}$};
    \node[position=140:{0.6cm} from v1] (aa) {$\cdots$};
    \node[position=40:{0.6cm} from vN] (cc) {$\cdots$};
    \path[draw,color=purduecoalgray,->] (u1) -- (cd3);
    \path[draw,color=purduecoalgray,->] (cd3) -- (uN);
    \path[draw,color=purduecoalgray,->] (v1) -- (v2);
    \path[draw,color=purduecoalgray,->] (v2) -- (cd2);
    \path[draw,color=purduecoalgray,->] (cd2) -- (vN);
    \path[draw,color=blue,->] (aa) -- (v1);
    \path[draw,color=red,->] (uN) -- (v2);
    \path[draw,color=tangocolordarkchameleon,->] (vN) -- (cc);
    \node[below=0.2cm of v2] (IDR) {\large $\mathsf{IDR}(G,\gamma)$};
\end{tikzpicture}
\caption{An example of the transformation $\idr(G,\gamma)$. We remark that if the red edge $(u,v)$ is the 2\nd incoming edge of $v$, then in $\idr(G,\gamma)$ we should add $(u_{\delta+\gamma},v_2)$ to $G'$.}
\figlab{fig:idr}
\end{figure}
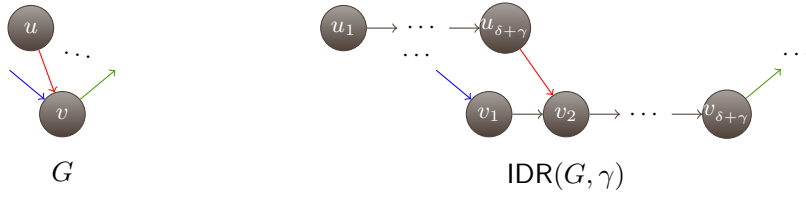

\section{Useful Theorems} \applab{useful}
We rely on the following results in our constructions and proofs. 
\begin{theorem}[\cite{C:AlwBlo16}]
\thmlab{thm:cc}
Let $G$ be a DAG with $N$ vertices and indegree $\delta$. 
If $G$ is $(e,d)$-reducible, then $\pcc(G)\le\underset{g\ge d}{\min}\, \left\{eN+\delta gN + \frac{N^2d}{g}\right\}$.
\end{theorem}

\begin{theorem}[\cite{EC:AlwBloPie17}]
\thmlab{thm:cc:upper}
Let $G$ be a DAG with $N$ vertices and indegree $\delta$. 
If $G$ is $(e,d)$-depth robust, then $\pcc(G)\ge ed$.
\end{theorem}

While \thmref{thm:cc} and \thmref{thm:cc:upper} are nice results that relate the pebbling complexities of $(e_1,d_1)$-reducible and $(e_2,d_2)$-depth robust graphs, these statements are ultimately misleading in that $d\le N$ and thus there will never be a gap between the pebbling complexities of the graphs. 

\begin{theorem}[\cite{Svensson12}]
\thmlab{thm:svensson}
For any integer $k\ge 2$ and constant $\eps>0$, given a DAG $G$ with $N$ vertices, it is Unique Games hard to distinguish between the following cases:

\begin{itemize}
\item
(Completeness): $G$ is $\left(\left(\frac{1-\eps}{k}\right)N,k\right)$-reducible.
\item
(Soundness): $G$ is $\left((1-\eps)N,N^{1-\eps}\right)$-depth robust.
\end{itemize}
\end{theorem}
\begin{definition}\deflab{def:extreme}
Given a parameter $0<\gamma<1$, a DAG $G=(V,E)$ is \emph{$\gamma$-extreme depth-robust} if $G$ is $(e,d)$-depth robust for any $e,d$ such that $e+d\le(1-\gamma)N$. 
\end{definition}
\begin{theorem}[\cite{EC:AlwBloPie18}]
For any fixed $0<\gamma<1$, there exists a constant $c_1>0$ such that for all integers $N>0$, there exists an $\gamma$-extreme depth robust graph $G$ with $N$ vertices and $\indeg(G),\outdeg(G)\le c_1\log N$. 
\end{theorem}

\begin{lemma}[\cite{EPRINT:BHKLXZ18}]
\lemlab{lem:superconc:cc:lower}
Let $G$ be an $(e,d)$-depth robust graph with $N$ vertices. 
Then
\[\pcc(\superconc(G))\ge\min\left(\frac{eN}{8},\frac{dN}{8}\right).\]
\end{lemma}

\section{Missing Proofs}\applab{missing}

\begin{remindertheorem}{\lemref{lem:indeg:abp}}
\thmindegreductionabp
\end{remindertheorem}
\begin{proof}[Proof of \lemref{lem:indeg:abp}.]
Alwen \etal\,\cite{EC:AlwBloPie17} show that $\idr(G,\gamma)$ is $(e,\gamma\cdot d)$-depth robust if $G$ is $(e,d)$-depth robust. 
It remains to show that $\idr(G,\gamma)$ is $(e,(\delta+\gamma)\cdot d)$-reducible if $G$ is $(e,d)$-reducible. 

Given an $(e,d)$-reducible graph $G=(V,E)$ of $N$ vertices, we use $[N]$ to represent the vertices of $G$ and let $G'=\idr(G,\gamma)$ so that the vertices of $G'$ can be associated with $[(\delta+\gamma)N]$. 
Let $S$ be a set of $e$ vertices in $G$ such that $\depth(V-S)<d$. 
Let $S'$ be a set of $e$ vertices in $G'$ so that $(\delta+\gamma)v\in S'$ for each vertex $v\in S$. 

Suppose, by way of contradiction, that there exists a path $P'$ of length $(\delta+\gamma)\cdot d$ in $G'-S'$. 
Observe that if $y,z\in G'$ such that $(\delta+\gamma)a+1\le y\le(\delta+\gamma)(a+1)$ and $(\delta+\gamma)b+1\le z\le(\delta+\gamma)(b+1)$ for integers $a<b$, then $y$ cannot be connected to $z$ unless $y=(\delta+\gamma)(a+1)$. 
Hence that if $P'$ contains vertex $u\in G'$ such that $(\delta+\gamma)c+1\le u\le(\delta+\gamma)(c+1)$ and $u$ is not one of the final $\delta+\gamma-1$ vertices of $P'$, then $(\delta+\gamma)(c+1)\in P'$. 
Thus by a simple Pigeonhole argument, there exists at least $d$ integers $j_1,j_2,\ldots,j_d$ such that $(\delta+\gamma)j_n\in P'$ for $1\le n\le d$ and moreover there exists an edge in $P'$ from each vertex $(\delta+\gamma)j_n$ to some vertex $w$ such that $(\delta+\gamma)(j_{n+1}-1)+1\le w\le (\delta+\gamma)j_{n+1}$ for $1\le n\le d-1$. 
However, this implies that $j_1,\ldots,j_d$ is a path in $G$ by construction of $\idr(G,\gamma)$. 
Moreover, since $(\delta+\gamma)v\in S'$ for each vertex $v\in S$, this implies that $j_1,\ldots,j_d$ is a path of length $d$ in $G-S$, which contradicts the assumption that $\depth(G-S)<d$. 
\end{proof}

\begin{remindertheorem}{\corref{cor:second:idr}}
\corsecondidr
\end{remindertheorem}
\begin{proof}[Proof of \corref{cor:second:idr}.]
Suppose $G'$ is a graph with $M$ vertices. 
By applying \lemref{lem:indeg:abp} to \thmref{thm:modified:sven} and setting $k=M^{2\eps}$ and $\gamma=M^{2\eps}-\delta$, then we start from a graph with $M$ vertices and end with a graph $G$ with $N=M^{1+2\eps}$ vertices or equivalently, $M=N^{1/(1+2\eps)}$. 
Thus, $G=\idr(G',\gamma)$ is $(e,d)$-reducible for $e=\frac{(1-\eps)M}{k}=\frac{1-\eps}{k}N^{1/(1+2\eps)}$ and $d=kM^{2\eps}=kN^{2\eps/(1+2\eps)}$. Since $e<\frac{M}{k}$, it is clearly the case that $G=\idr(G',\gamma)$ is $(e_1,d_1)$-reducible for $e_1=\frac{M}{k}>e$ and $d_1=d=kN^{2\eps/(1+2\eps)}$ as we delete more nodes and the depth reducibility guarantees the same upper bound of the remaining depth.
On the other hand, $\idr(G',\gamma)$ is $(e_2,d_2)$-depth robust for $e_2=(1-\eps)M=(1-\eps)N^{1/(1+2\eps)}$, while $d_2=\gamma M^{1-\eps}=(M^{2\eps}-\delta) M^{1-\eps}$. 
By \thmref{thm:modified:sven}, $\delta=\O{M^{\eps} \log^2 M}$ so that for sufficiently large $M$, $d_2=0.9 M^{1+\eps}=0.9 N^{(1+\eps)/(1+2\eps)}$.  
\end{proof}

\begin{lemma}
\lemlab{lem:overlay:ed}
If $G$ is $(e,d)$-reducible, then $\superconc(G)$ is $\left(e+\frac{N}{d},2d+\log(42N)\right)$-reducible, where $N$ is the number of vertices in $\superconc(G)$.  
\end{lemma}
\begin{proof}
Let $G=(V,E)$ be a $(e,d)$-reducible DAG with $N$ vertices. 
Let $G'=\superconc(G)$ and suppose $G'$ has $M$ vertices, which we designate $[M]$. 
Thus, there exists a set $S\subseteq V$ such that $|S|\le e$ and $\depth(G-S)<d$. 
Let $T$ be the set of $\frac{N}{d}$ vertices $\{M,M-d,M-2d,\ldots,M-N+d\}$, so that $T\subseteq\outp(G')$.  
We claim $\depth(G'-S-T)<2d+\log(42N)$. 

Suppose by way of contradiction that there exists a path $P$ in $G'-S-T$ of length at least $2d+\log(42N)$. 
By \lemref{lem:pippenger}, the depth of any path from an input node to an output vertex is at most $\log(42N)$. 
Moreover, all edges added in the superconcentrator overlay are either between input vertices or two output vertices. 
Hence, then at least $2d$ vertices of $P$ have to lie in either the first $N$ vertices or the last $N$ vertices of $G'$. 
Because $P$ does not contain vertices of $T$, there is no path of length of length $d$ in the last $N$ vertices of $G'$, so there must be a path of length $d$ in the first $N$ vertices of $G'$, which contradicts $\depth(G-S)<d$.

Therefore, $G'$ is $\left(e+\frac{N}{d},2d+\log(42N)\right)$-reducible.
\end{proof}

From \lemref{lem:overlay:ed} we immediately  obtain an upper bound on the pebbling complexity of $\superconc(G)$ by applying \thmref{thm:cc} to \lemref{lem:overlay:ed}. However, the upper bound is not as strong as we would like.  
\begin{corollary}\corlab{cor:sc(G)naive}
Let $G$ be an $(e,d)$-reducible graph with $N$ vertices. 
Then
\[\pcc(\superconc(G))\le\underset{g\ge d}{\min}\,\left(e+\frac{N}{d}\right)42N+2g(42N)+\frac{42N}{g}\left(2d+\log(42N)\right)42N.\]
\end{corollary}

\end{document}